\documentclass[11pt, letterpaper]{article}

\usepackage{macros}
\usepackage[font=small]{caption}

\begin{document}

\title{Perturbation Resilient Clustering for $k$-Center \\
and Related Problems via LP Relaxations\thanks{Department of Computer Science,
University of Illinois, Urbana-Champaign, IL 61801. Email: {\tt \{chekuri, sgupta49\}@illinois.edu.} Work on this paper supported in part by
NSF grants CCF-1319376 and CCF-1526799.} \\
{\small (Full Version)}
}
\author{Chandra Chekuri \and Shalmoli Gupta} %
\date{\today}
\maketitle

\begin{abstract}
  We consider clustering in the perturbation resilience model that has
  been studied since the work of Bilu and Linial~\cite{BiluL10} and
  Awasthi, Blum and Sheffet~\cite{AwasthiBS12}. A clustering instance
  $\inst$ is said to be $\alpha$-perturbation resilient if the optimal
  solution does not change when the pairwise distances are modified by
  a factor of $\alpha$ and the perturbed distances satisfy the metric
  property --- this is the metric perturbation resilience property 
  introduced in \cite{AngelidakisMM17} and a weaker requirement than prior 
  models. We make two high-level contributions.

  \begin{itemize}
  \item We show that the natural LP relaxation of \kcenter and
    asymmetric \kcenter is integral for $2$-perturbation resilient
    instances. We belive that demonstrating the goodness of standard
    LP relaxations complements existing results
    \cite{BalcanHW16,AngelidakisMM17} that are based on new algorithms
    designed for the perturbation model.
  
\item We define a simple new model of perturbation resilience for
  clustering with \emph{outliers}. Using this model we show that the
  unified MST and dynamic programming based algorithm proposed in
  \cite{AngelidakisMM17} exactly solves the clustering with outliers
  problem for sevearl common center based objectives (like \kcenter,
  \kmeans, \kmedian) when the instances is $2$-perturbation resilient.
  We further show that a natural LP relxation is integral for
  $2$-perturbation resilient instances of \kcenter with outliers.

  \end{itemize}

\end{abstract}

\newpage

\section{Introduction}
\label{sec:intro}
Clustering is an ubiquitous task that finds applications in numerous
areas since it is a basic primitive in data analysis.  Consequently,
clustering methods are extensively studied in many scientific
communities and there is a vast literature on this topic. In a typical
clustering problem the input is a set of points with a notion of
similarity (also called \emph{distance}) between every pair of points,
and a parameter $k$, which specifies the desired number of
clusters. The goal is to partition the points into $k$ clusters such
that points assigned to the same cluster are similar. One way to
obtain this partition is to select $k$ centers and then assign each
point to the nearest center. The quality of the clustering can be
measured in terms of an objective function. Some of the popular and
commonly studied ones are \kmedian (sum of distances of points to
nearest center), \kmeans (sum of squared distances of points to
nearest center), and \kcenter (maximum distance between a point to its
nearest center). These are center-based objective functions.  Unlike
some applications in Operations Research, in many clustering problems
in data analysis, the objective function is a proxy to identify the
clusters and the actual value of the objective function is not necessarily
meaningful. Clustering is often considered in the presence of outliers.
In this setting the goal is to find the best clustering of the
input after removing (at most) a specified number (or fraction) of points ---
this is useful in practice when the input data is noisy.

Most of the natural optimization problems that arise in 
clustering turn out to be \NPHard. Extensive work exists on
approximation algorithm design as well as heuristics. Although
clustering and its variants are intractable in the worst case, various
heuristic based algorithms like Lloyd's, K-Means++ perform very well
in practice and are routinely used --- at the same time some of these
heuristics have poor worst-case approximation performance.  On the
other hand algorithms designed for worst-case approximation bounds may
not work well in practice or may not be sufficiently fast for large
data sets.  To bridge this gap between theory and practice, there has
been an increasing emphasis on \emph{beyond worst case analysis}.
Several models have been proposed to understand real-world instances
and why they may be computationally easier.  One such model is based
on the notion of \emph{instance stability}.  This is based on the
assumption that typical instances have a clear underlying optimal
clustering (also known as \emph{ground-truth clustering}) which is
significantly better than all other clusterings, and remains the
same under small perturbations. 

The notion of stability/perturbation resilience was formalized in the
work of Bilu and Linial \cite{BiluL10} initially for Max Cut, and by
Awasthi, Blum and Sheffet \cite{AwasthiBS12} for clustering. For
clustering problems, an instance $\inst$ is said to be
$\alpha$-perturbation resilient for some $\alpha > 1$ if the optimum
clustering remains the same even if pairwise distances between points
are altered by a multiplicative factor of at most
$\alpha$. Intuitively, $\alpha$ determines the degree of resilience of
the instance, with a higher value translating to more structured, and
separable instances. In the past few years, there has been increasing
interest in understanding stable/perturbation-resilient
instances. After several papers
\cite{AwasthiBS12,BalcanL12,BalcanHW16}, a recent result by
Angelidakis, Makarychev and Makarychev \cite{AngelidakisMM17} showed
that $2$-perturbation resilient instances of several clustering
problems with center based objectives (which includes \kmedian,
\kcenter, \kmeans) can be solved exactly in polynomial time. For
\kcenter finding the optimum solution for $(2-\delta)$-perturbation
resilient instances is \NPHard \cite{BalcanHW16}.  One criticism of
perturbation resilience for clustering was the assumption in some
earlier works that the optimum clustering remains stable under
perturbation of the original metric $d$ even when perturbed distance $d'$ itself may
not be a metric. Interestingly the results of \cite{AngelidakisMM17}
hold even under the weaker assumption of \emph{metric} perturbation
resilience, which constrains the perturbed pairwise distances to be a
metric. The results in \cite{AngelidakisMM17} are based on a simple
and unified algorithm that computes the MST $T$ of the given set of
points and then applies dynamic programming on $T$ to find the
clusters; it is only in the second step that the specific objective
function is used. We believe that empirically evaluating the
performance of this algorithm, and related heuristics, on real-world
data is an interesting avenue and plan to study it. Our work in this
paper is motivated by the existing work and several interrelated
questions on theoretical concerns, that we discuss next.

One of the objectives in beyond-worst-case analysis is to explain the
empirical success of existing algorithms and mathematical programming
formulations. For stable instances of Max-Cut and Minimum Multiway
Cut, convex relaxations are known to be integral for various bounds on
the perturbation parameter \cite{MakarychevMV14,AngelidakisMM17}. In
the context of \kmedian and \kmeans Awasthi {\etal}
\cite{AwasthiCKS15} showed that if the data is generated uniformly at
random from $k$ unit balls with well-separated centers, convex
relaxations (linear and semi-definite) give an optimal itegral
solution under appropriate separation conditions on the
centers. However, for perturbation resilient clustering instances not
much is known about the the natural LP relaxations.  This raises a
natural question.

\begin{question}
  \label{q:lp}
  Are the natural LP relaxations for $2$-metric perturbation resilient
  instances of clustering problems integral?
\end{question} 

There are several advantages in proving that well-known relaxations
are integral. First, they provide evidence of the goodness of the
relaxation; often these relaxations also have worst-case approximation
bounds.  Second, when the relaxation does not give an integral
solution for a given instance we can deduce that the instance is
\emph{not} perturbation resilient.

As we remarked, one major takeaway from the paper of Angelidakis
{\etal} \cite{AngelidakisMM17}, apart from its strong theoretical
results, is the simple and unified algorithm that they propose which may
lead to an effective heuristic. In real-world data there is often
noise, and it would be useful to develop algorithms 
in the more general setting of clustering with outliers.
This leads us to the question,

\begin{question}
  \label{q:outliers}
  Is there any stability model under which the 
  algorithm proposed by \cite{AngelidakisMM17} gives optimal
  solution for the problem of clustering with outliers?
\end{question}

We remark that even for instances without outliers, removing a small
fraction of the points can lead to a residual instance which has better
stability parameters than the initial one. Thus, clustering with outlier
removal is relevant even when there is no explicit noise.

\subsection{Our Results}
In this paper we address the preceding questions and obtain the
following results. 

\begin{itemize}
\item We show that a natural LP relaxation for
  \kcenter has an optimum integral solution for
  $2$-metric-perturbation resilience instances\footnote{Although the
    LP provides a $2$-approximation it is not immediate that it would
    be exact for perturbation-resilience instances}. Thus, when running the LP on a clustering instance, either we are guaranteed to have found the optimal solution (if the LP solution is integral), or we are guaranteed the solution is not $2$-perturbation resilient (if the LP solution is not integral). The previous algorithms of Angelidakis \etal \cite{AngelidakisMM17}, and Balcan \etal \cite{BalcanHW16} do not have this guarantee, and could be arbitrarily bad if the instance is not \pr{2}.

\item Motivated by the work of \cite{BalcanHW16} we consider the
  \emph{asymmetric} \kcenter (\asymkcenter) problem. We show that a
  natural LP relaxation has an optimum integral solution for
  $2$-\emph{metric}-perturbation resilient instances\footnote{In the
    asymmetric setting the perturbed distances should satisfy triangle
    inequality but symmetry is not required.}. For \asymkcenter the
  worst-case integrality gap of the LP relaxation is known to be
  $\Theta(\log^* k)$ \cite{Archer01,ChuzhoyGHKKKN05}. Previously
  \cite{BalcanHW16} described a specific combinatorial algorithm that
  outputs an optimum solution for $2$-perturbation resilient
  instances. We obtain it via the LP relaxation in the weaker
  metric perturbation model.

\item We define a simple model of perturbation resilience for
  clustering with \emph{outliers}. It is a clean extension of the
  existing perturbation resilience model. We show that under this new
  model, a modification of the algorithm of Angelidakis \etal
  \cite{AngelidakisMM17} gives an exact solution for the outliers
  problem (for \kmedian, \kmeans, \kcenter and outer $\ell_p$ based
  objectives). This algorithm may lead to an interesting heuristic for
  clustering (noisy) real-world instances. We also show that for a
  $2$-perturbation resilient instance of \kcenter with outliers, a
  natural LP relaxation has an optimum integral solution.
\end{itemize}

Our results show the efficacy of LP relaxations for \kcenter and its
variants. We also demonstrate, via a natural model, that the
interesting algorithm from \cite{AngelidakisMM17} extends to handle
outliers.  Perturbation resilience appears to be a simple definition
but it is hard to pin down its precise implications. Prior work
demonstrates that observations and algorithms that appear simple in
retrospect have not been easy to find. For \kcenter and \asymkcenter
we work with notion of perturbation resilience under Voronoi clusterings
as was done in \cite{BalcanHW16}; this is the more restrictive version.
See \refsection{prelim} for the formal definitions.


We would like to understand the integrality gap of the natural LP
relaxations for perturbation resilient instances of \kmedian and
\kmeans. We believe that the following open question is quite
interesting to resolve.

\begin{question}
  Is there a fixed constant $\alpha$ such that the natural LP
  relaxation for \kmedian (similarly \kmeans) has an integral optimum
  solution for every $\alpha$-perturbation resilient
  instance\footnote{It may be possible to answer this question in the
    positive if we additionally assume that the optimum
    clusters are balanced in terms of number of points. However, we feel
    that such an assumption does not shed light on the structure of
    perturbation resilient instances that are not balanced.}?
\end{question}

\subsection{Related Work}
There is extensive related work on clustering topics.
Here we only mention some closely related work. 
\vspace{-3mm}
\paragraph{Clustering.} 
For both \kcenter and asymmetric \kcenter tight approximation bounds
are known. For \kcenter, already in the mid $1980$'s Gonzales
\cite{Gonzalez85} and Hochbaum \& Shmoys \cite{HochbaumS85} had
developed remarkably simple $2$-approximation algorithms, which are
in fact tight. Approximating asymmetric \kcenter is significantly
harder. Panigrahy and Vishwanathan \cite{PanigrahyV98} designed an
elegant $O(\log^* n)$ approximation algorithm, which was subsequently
improved by Archer \cite{Archer01} to $O(\log^* k)$.  Interestingly,
the result is asymptotically tight \cite{ChuzhoyGHKKKN05}.

For \kmeans and \kmedian --- arguably the two most popular clustering
problems --- there is a long line of research (see \cite{BlomerL0S16}
for a survey on \kmeans). The first constant factor approximation for
the \kmedian problem was given by Charikar {\etal}
\cite{CharikarGTS99}, and the current best-known is a $2.675$
approximation by Byrka {\etal} \cite{ByrkaPRST17}; and it is \NPHard
to do better than $1 + 2/e \approx 1.736$ \cite{JainMMSV03}. For
\kmeans the best approximation known is $6.357$
\cite{AhmadianNSW17}. The \kmeans problem is widely used in practice
as well, and the commonly used algorithm is Lloyd's algorithm, which
is a special case of the EM algorithm \cite{Lloyd82}. While there is
no explicit approximation guarantee of the algorithm, it performs
remarkably well in practice with careful seeding \cite{ArthurV07}
(this heuristic is called K-Means++).

\paragraph{Clustering with Outliers.} 
The influential paper by Charikar {\etal} \cite{Charikarkmn01}
initiated the work on clustering with outliers and other robust
clustering problems. For \kcenter with outliers, they gave a greedy
$3$-approximation algorithm. Further, for \kmedian with outliers they
gave a bicriteria approximation algorithm, which achieves an
approximation ratio of $4(1+ \epsilon)$, violating the number of
outliers by a factor of $(1+\epsilon)$. The first constant factor
approximation algorithm for this problem was given by Chen (the
constant is not explicitly computed) \cite{Chen08}. Very recently,
Krishnaswamy {\etal} \cite{KrishnaswamyLS17} proposed a generic
framework for clustering with outliers. It improves the results of
Chen and gives the first constant factor approximation for \kmeans
with outliers. However, the algorithm does not appear 
suitable for practice in its current form (See \cite{Aggarwal13} for details on algorithms used in practice for clustering with outliers). 

\paragraph{Perturbation Resilience.}
The notion of perturbation resilience was introduced by Bilu and
Linial \cite{BiluL10}.  They originally considered it for the Max Cut
problem, designing an exact polynomial time algorithm for
$O(n)$-stable instances \footnote{They used this name to denote
  perturbation resilient instances of Max Cut} of Max Cut. It was
later improved to $O(\sqrt{n})$-stable instances \cite{BiluDLS13}, and
finally Makarychev {\etal} gave a polynomial time exact algorithm for
$O(\sqrt{\log n}\cdot \log \log n)$-stable instances
\cite{MakarychevMV14}.

The definition of perturbation resilience naturally extends to
clustering problems. Awasthi, Blum, and Sheffet \cite{AwasthiBS12}
presented an exact algorithm for solving $3$-perturbation resilient
clustering problems with \emph{separable center based objectives}
(s.c.b.o) --- this includes \kmedian, \kmeans, \kcenter. This result
was later improved by Balcan and Liang \cite{BalcanL12}, who gave an
exact algorithm for clustering with s.c.b.o under $(1 +
\sqrt{2})$-perturbation resilience. Specifically for \kcenter and
asymmetric \kcenter, Balcan, Haghtalab, and White \cite{BalcanHW16}
designed an algorithm for $2$-perturbation resilient instances. In
fact, for \kcenter they gave a stronger result, that any
$2$-approximation algorithm for \kcenter can give an optimal solution
for $2$-perturbation resilient instances. They also showed the results
are essentially tight unless \NP = \RP \footnote{They showed, unless
  \NP = \RP, no polynomial-time algorithm can solve \kcenter under $(2
  - \epsilon)$-approximation stability, a notion that is stronger than
  perturbation resilience}. Recently, Angelidakis \etal
\cite{AngelidakisMM17}, gave an unifying algorithm which gives exact
solution for $2$-perturbation resilient instances of clustering
problems with center based objectives. In fact, their algorithms work
under metric perturbation resilience, which is a weaker
assumption. Perturbation resilience has also been studied in various
other contexts, like TSP, Minimum Multiway Cut, Clustering with
min-sum objectives \cite{BalcanL12,MakarychevMV14,MihalakSSW11}.

\paragraph{Robust Perturbtion Resilience.} 
Perturbation resilience requires optimal solution to remain unchanged
under any valid perturbation. Balcan and Liang \cite{BalcanL12}
relaxed this condition slightly, and defined $(\alpha,
\epsilon)$-perturbation resilience (or robust perturbation
resilience), in which at most $\epsilon$ fraction of the points can
change their cluster membership under any $\alpha$-perturbation. They
gave a near optimal solution for \kmedian under $(4,
\epsilon)$-perturbation resilience, when the clusters are not too
small. Further, for \kcenter and asymmetric \kcenter efficient
algorithms are known for $(3, \epsilon)$-perturbation resilient
instances, assuming mild size lower bound on optimal clusters
\cite{BalcanHW16}.

\paragraph{Other Stability Notions.} 
Several other stability models, and separation conditions have also
been studied to better explain real-world instances. In a seminal
paper Ostrovsky, Rabani, Schulman, and Swamy \cite{OstrovskyRSS06}
considered \kmeans instances where the cost of clusterng using $k$ is
clusters is much lower than $k-1$ clusters. They showed, that popular
K-Means++ algorithm achieves an $O(1)$-approximation for these
instances. Subsequently there has been series of work many other
models like approximation stability \cite{BalcanBG09}, agnostic clustering \cite{BalcanHS09}, distribution
stability \cite{AwasthiBS10,CohenAddadS17}, spectral separability
\cite{KumarR10,AwasthiS12,CohenAddadS17}, and more recently on
additive perturbation stability \cite{VijayaraghavanDW17}.

\medskip
\noindent {\bf Organization:} The rest of the paper is organized as
follows: in \refsection{prelim} we formally define the clustering
problems and perturbation resilience; in \refsection{kc} we prove that
any 2-approximation algorithm gives optimal solution for a
$2$-perturbation resilient \kcenter instance, further we show that the
natural LP is integral; in \refsection{asymkc} we show that even for
asymmetric \kcenter the natural LP relaxation is integral under
$2$-perturbation resilience; in \refsection{kco} we prove the
integrality of LP for $2$-perturbation resilient \kcenter with
outliers instance; finally in \refsection{kmo} we show present a
dynamic programming based algorithm which exactly solves \kmedian with outlers (and
also \kcenterout, \kmeansout ) under $2$-perturbation resilience.

\section{Preliminaries}\labelsection{prelim}
\subsection{Definitions \& Notations}
In this section we formally define the clustering problems and perturbation resilience.

\paragraph{Clustering.} 
An instance $\inst$ of a clustering problem is defined by the tuple 
$(V, d, k)$, where $V$ is a set of $n$ points, 
$\maps{d}{V \times V}{\realsnonnegative}$ is a metric distance function, and $k$ is an 
integer parameter. The goal is to find a set of $k$ distinct points 
$S = \curlyof{c_1, \ldots, c_k} \subseteq V$ called \emph{centers} such that an objective
function defined over the points is optimized. The objective
function, also known as clustering cost, can be defined in various ways, and depends on the
problem in hand. Here, we are interested
in the {\kmedian}, and {\kcenter} objectives. Given a set of centers $S = \curlyof{c_1, \ldots, c_k}$
these objectives are defined as follows:
\begin{align*}
\text{(\kmedian)} \quad \quad \cost{d}{S} 
&= \sum_{u \in V} d(S, u) \\
\text{(\kmeans)} \quad \quad \cost{d}{S} 
&= \sum_{u \in V} d^2(S, u) \\
\text{(\kcenter)} \quad \quad \cost{d}{S} 
&= \max_{u \in V} d(S, u)
\end{align*}
where $d(S, u) = \Min_{i \in \curlyof{1, \ldots, k }} d(c_i, u)$.

The Voronoi partition induced by the centers, gives a natural way of clustering the input point set. 
In fact, the inherent goal of clustering is to uncover the underlying partitioning of points, and
one expects with correct choice of distance modeling, "$k$", and objective function, 
the Voronoi partition induced by the optimal set of centers will reveal the underlying clustering.
Throughout this paper, whenever we mention \emph{optimal clustering}, we indicate the Voronoi partition
corresponding to the optimal set of centers. Thus with this dual view of the clustering problem,
given a set of centers $S = \curlyof{c_1, \ldots, c_k}$, and corresponding Voronoi partition 
$\mathcal{C} = \curlyof{C_1, \ldots, C_k}$, the clustering cost can be rewritten as:
\begin{align*}
\text{(\kmedian)} \quad \quad \cost{d}{\mathcal{C}, S} 
&= \sum_{i = 1}^k \sum_{u \in C_i} d(c_i, u) \\
\text{(\kmeans)} \quad \quad \cost{d}{\mathcal{C}, S} 
&= \sum_{i = 1}^k \sum_{u \in C_i} d^2(c_i, u) \\
\text{(\kcenter)} \quad \quad \cost{d}{\mathcal{C}, S}
&= \max_{i \in \curlyof{1, \ldots, k }} \max_{u \in C_i} d(c_i, u)
\end{align*}

So far, in the clustering problem instance, we considered the distance function $d$ to be a metric.
However, this may not always be the case. Specifically, for the \kcenter objective, a generalization
 which is also studied is the Asymmetric \kcenter problem (\asymkcenter), where the distance
 function $d$ in the input instance $\inst = (V, d, k)$ is an asymmetric distance function. In other words, $d$
 obeys triangle inequality, but not symmetry.
That is $d(u, v) \leq d(u, w) + d(w, v)$ for all $u, v, w \in V$. However $d(u, v)$
may be not be same as $d(v, u)$. The objective is the \kcenter objective, but because the distance is assymetric,
order matters - we define the cost in terms of distance \emph{from the center to the points} i.e. given a center
$c$ and a point $u$, $d(c, u)$ is used to define cost. To reiterate,
given a set of centers $S = \curlyof{c_1, \ldots, c_k}$ and corresponding Voronoi partition (w.r.t $d(c_i, u)$)
$\mathcal{C} = \curlyof{C_1, \ldots, C_k}$, the clustering cost is:
\begin{align*}
\text{(\asymkcenter)} \quad \quad \cost{d}{\mathcal{C}, S} 
&= \max_{i \in \curlyof{1, \ldots, k }} \max_{u \in C_i} d(c_i, u)
\end{align*}

\paragraph{Clustering with Outliers.} 
An instance $\inst$ of a clustering with outliers problem is defined by the tuple 
$(V, d, k, z)$, where $V$ is a set of $n$ points, 
$\maps{d}{V \times V}{\realsnonnegative}$ is a metric distance function, and $k, z$ are
integer parameters. The goal is to identify $z$ points $Z \subseteq V$ as \textit{outliers} and partition the remaining $V \setminus Z$ points into $k$ clusters such that the clustering cost 
is minimized. Formally, given a set of outliers $Z$, a set of centers 
$S = \curlyof{c_1, \ldots, c_k} \subseteq V \setminus Z$,
and a Voronoi partition of $V \setminus Z$, $\mathcal{C} = \curlyof{C_1, \ldots, C_k}$ induced
by $S$, the clustering cost is defined as:
\begin{align*}
\text{(\kmedianout)} \quad \quad \cost{d}{\mathcal{C}, S; Z} 
&= \sum_{i = 1}^k \sum_{u \in C_i} d(c_i, u) \\
\text{(\kmeansout)} \quad \quad \cost{d}{\mathcal{C}, S; Z} 
&= \sum_{i = 1}^k \sum_{u \in C_i} d^2(c_i, u) \\
\text{(\kcenterout)} \quad \quad \cost{d}{\mathcal{C}, S; Z}
&= \max_{i \in \curlyof{1, \ldots, k }} \max_{u \in C_i} d(c_i, u)
\end{align*}

\paragraph{Perturbation Resilience.}
A clustering instance $\inst = (V, d, k)$ is $\alpha$-metric perturbation resilient ($\pr{\alpha}$) for a given
objective function, if for any metric 
\footnote{In case of \asymkcenter, we consider perturbations in which $d'$ obeys triangle inequality, but not symmetry} 
distance function $\maps{d'}{V \times V}{\realsnonnegative}$, 
such that for all $u, v \in V$, $\frac{d(u,v)}{\alpha} \leq d'(u, v) \leq d(u,v)$, 
the unique optimal clustering of $\inst' = (V, d', k)$ 
is identical to the unique optimal clustering of $\inst$. 

Note that after perturbation the optimal centers may change, however for the instance to be perturbation
resilient, the optimal clustering i.e. Voronoi partition induced by the optimal centers
must stay the same. Unless otherwise noted, for the rest of the paper 
$\alpha$-perturbation resilient indicates metric perturbation resilience. 

\paragraph{Outlier Perturbation Resilience.}
A clustering with outliers instance $\inst = (V, d, k, z)$ is $\alpha$-metric 
outlier perturbation resilient ($\opr{\alpha}$) for 
a given objective function, if for any metric distance function 
$\maps{d'}{V \times V}{\realsnonnegative}$, such that for all $u, v \in V$,
$\frac{d(u,v)}{\alpha} \leq d'(u, v) \leq d(u,v)$, the unique optimal clustering and outliers of 
$\inst' = (V, d', k, z)$ are identical to the optimal solution of $\inst$. 

It is easy to see, if a clustering with outliers instance $(V, d, k, z)$ with unique optimal clusters $\mathcal{C}$ and
outliers $Z$ is $\opr{\alpha}$, then the clustering instance $(V \setminus Z, d, k)$
is $\pr{\alpha}$.

\paragraph{Notation.} For integer,
$k$, let $[k] = \curlyof{1, \ldots, k}$. Throughout, we use $V$ to denote the input set of points, and $n$ is
the number of points.  For any clustering instance (including outlier
instances), $S = \curlyof{c_1, \ldots, c_k}$ denotes an optimal set of
centers, and $\mathcal{C} = \curlyof{C_1, \ldots, C_k}$ denotes the
corresponding Voronoi partition, which we call optimal
clusters. Further, for a point $p \in C_i$, we often interchangebly
use the terms, $p$ is \emph{assigned}/\emph{belongs} to center $c_i$
or cluster $C_i$. For a clustering with outlier instance, $Z$ denotes
the optimal set of outliers. In case of \kcenter, we refer to the
optimal clustering cost as \emph{optimal radius}, and denote it as
$\optradius{d}$.

\subsection{Some useful lemmas}
Here we state some intuitive and useful lemmas regarding \kcenter and
and \asymkcenter instances. The proofs of these lemmas are fairly
simple and can be found in \refappendix{gen}.

Recall, in the definition of perturbation resilience, we insisted that
the optimal $k$ clustering of the perturbed instance $\inst'$ has to
be same as the optimal $k$ clustering of the original instance. It is
not hard to show, that for \asymkcenter (and also for \kcenter), if a
$k-1$ clustering of $\inst'$ exists whose cost ist at most the optimal
cost of $k$ clustering, then the instance is not perturbation
resilient. Formally,

\begin{lemma}\labellemma{voronoi-kc}
Consider any {\asymkcenter} instance $\inst = (V, d, k)$. Let $S = \curlyof{c_1, \ldots, c_k}$
be an optimal set of centers, and $\mathcal{C} = \curlyof{ C_1, \ldots C_k }$ 
be the corresponding optimal clustering. The optimal radius is $\optradius{d}$. 
Suppose there exists a set of $k-1$ centers $S' = \curlyof{c_1', \ldots, c_{k-1}'}$,
inducing the Voronoi partition $\mathcal{C'} = \{ C_1', \ldots, C_{k-1}' \} $,
with cost $\cost{d}{\mathcal{C'}, S'} \leq \optradius{d}$. Then, the optimal 
clustering $\mathcal{C}$ is not unique.
\end{lemma}

One common technique we use in multiple arguments, is perturbing the input instance
in a structured way. The next two lemmas are related to that.

\begin{lemma}\labellemma{d-perturb}
Consider a set of points $V$, and let $d$ be an asymmetric distance function defined over $V$.
Let $G$ be a complete directed graph on vertices $V$. The edge lengths in graph $G$ are given
by the function $\ell$, where for any edge $(u, v)$, 
$\frac{d(u, v)}{2} \leq \ell(u, v) \leq d(u, v)$. Then the distance function $d'$, defined as 
the shortest path distance in graph $G$ using $\ell$, is a metric\footnote{satsifies triangle inequality, and not necessarily symmetry} $2$-perturbation of $d$.
\end{lemma}

\begin{lemma}\labellemma{opt-unchanged}
Consider an \asymkcenter instance $\inst = (V, d, k)$, and let $\mathcal{C}$ be 
the optimal clustering and $\optradius{d}$ be the optimal radius.
Let $G$ be a complete directed graph over vertex set $V$. The edge lengths in graph $G$ are given
by the function $\ell$, where 
(1) for a subset of edges $E'$, $\ell(u, v) = \min \curlyof{d(u, v), \optradius{d}}$;
(2) for every other edge, $\ell(u, v) = d(u, v)$. Suppose $d'$ is defined as the shortest
path distance in graph $G$ using $\ell$. Consider the \asymkcenter instance $\inst' = (V, d', k)$,
let $\optradius{d'}$ be the
optimal radius. If $\mathcal{C}$ is an optimal clustering in $\inst'$, 
then $\optradius{d} = \optradius{d'}$.
\end{lemma}

The undirected versions of \reflemma{d-perturb} and \reflemma{opt-unchanged}
are as follows:
\begin{lemma}\labellemma{d-perturb-undirected}
Consider a set of points $V$, and let $d$ be a metric defined over $V$.
Let $G$ be a complete undirected graph on vertices $V$. The edge lengths in graph $G$ are given
by the function $\ell$, where for any edge $(u, v)$, 
$\frac{d(u, v)}{2} \leq \ell(u, v) \leq d(u, v)$. Then the distance function $d'$, defined as 
the shortest path distance in graph $G$ using $\ell$, is a metric $2$-perturbation of $d$.
\end{lemma}

\begin{lemma}\labellemma{opt-unchanged-undirected}
Consider a \kcenter instance $\inst = (V, d, k)$, and let $\mathcal{C}$ be 
the optimal clustering and $\optradius{d}$ be the optimal radius.
Let $G$ be a complete undirected graph over vertex set $V$. The edge lengths in graph $G$ are given
by the function $\ell$, where 
(1) for a subset of edges $E'$, $\ell(u, v) = \min \curlyof{d(u, v), \optradius{d}}$;
(2) for every other edge, $\ell(u, v) = d(u, v)$. Suppose $d'$ is defined as the shortest
path distance in graph $G$ using $\ell$. Consider the \kcenter instance $\inst' = (V, d', k)$,
let $\optradius{d'}$ be the
optimal radius. If $\mathcal{C}$ is an optimal clustering in $\inst'$, 
then $\optradius{d} = \optradius{d'}$.
\end{lemma}

\section{\texorpdfstring{\kcenter}{k-center} under Perturbation Resilience}\labelsection{kc}

In this section, we show that the
natural LP relaxation for a $2$-perturbation resilient \kcenter has an
integral optimum solution. To this end consider the result of Balcan \etal \cite{BalcanHW16} --- \emph{any}
$2$-approximation algorithm for \kcenter finds the optimal clustering for a $2$-perturbation resilient instance. They proved this result under the stronger definition of non-metric perturbation resilience, which was subsequently extended to metric perturbation resilience in an unpublished follow-up paper \cite{BalcanW17}. Formally, the result is as follows:

\begin{theorem}\labeltheorem{any-2}
  Let $\mathcal{A}$ be an arbitrary $2$-approximation algorithm for
  \kcenter.  Consider a $2$-perturbation resilient \kcenter instance
  $\inst = (V, d, k)$.  Let $\mathcal{C} = \curlyof{ C_1, \ldots C_k}$
  be the unique optimum clustering. Suppose $B$ is the set of centers
  returned by algorithm $\mathcal{A}$ when invoked on $\inst$. Then,
  the Voronoi partition induced by $B$ gives the optimal clustering
  $\mathcal{C}$.
\end{theorem}


\begin{proof}
  Let $\optradius{d}$ denote the optimum solution value for the given
  instance.  Let $\mathcal{C'}$ be a Voronoi partition induced by $B$.  In clustering
  $\mathcal{C'}$, for each point $p \in V$, let $c(p)$ be the center
  in $B$ it is assigned to i.e. $d(c(p), p) \leq d(B \setminus
  \setof{c(p)}, p)$.  We define a distance function $d'$ which is a
  metric $2$-perturbation of $d$: consider the complete graph $G$ on
  vertices $V$.  Let $E' = \curlyof{(c(p), p): p \in V}$. The edge
  lengths in graph $G$ are given by the function $\ell$, where for any
  edge $(u, v)$,
\[
\ell(u, v) = 
  \begin{cases} 
    \min \curlyof{d(u, v), \optradius{d}}  \quad& (u,v) \in E' \\
    d(u, v) \quad& \text{otherwise}
  \end{cases}
\]
For any pair of points $u, v$, the distance $d'(u, v)$ is the shortest
path distance between $u$ and $v$ in graph $G$, using $\ell$.

\begin{obs}\labelobs{kc-2-perturbed}
$d'$ is a metric $2$-perturbation of $d$.
\end{obs}
\begin{proof}
 Since algorithm $\mathcal{A}$ returns a $2$-approximate
  solution, for each point $p \in V$, $d(c(p), p) \leq 2 \cdot
  \optradius{d}$. Therefore, for any $(u, v) \in E'$, $\ell(u, v) =
  \min \curlyof{d(u, v), \optradius{d}} \geq \frac{d(u, v)}{2}$. For
  any other $(u, v)$, by definition $\ell(u, v) = d(u, v)$. In other
  words, for any $u, v$, we have $\frac{d(u, v)}{2} \leq \ell(u, v)
  \leq d(u, v)$. As stated in \reflemma{d-perturb-undirected}, $d'$
  defined as the shortest path distance in graph $G$ with edge lengths satisfying $\frac{d(u, v)}{2} \leq \ell(u, v) \leq d(u, v)$, is 
  a metric $2$-perturbation of $d$.
  \end{proof}

\begin{obs}\labelobs{kc-d-1}
  For any $p \in V$, $d'(c(p), p) \leq \min \curlyof{d(c(p), p),
    \optradius{d}}$. Further, $d'(p, B \setminus \setof{c(p)}) \geq \min
  \curlyof{d(c(p), p), \optradius{d}}$.
\end{obs}
\begin{proof}
  The first claim follows immediately from the fact $d'(c(p), p) \leq
  \ell(c(p), p)$.  For the second claim, consider any $s \in B
  \setminus \setof{c(p)}$.  Let $P$ be an arbitrary $s \leadsto p$
  path.  If $P \bigcap E' = \emptyset$, then by triangle inequality
  $\ell(P) = \sum_{ e \in P} d(e) \geq d(s, p) \geq d(c(p), p)$.
  Otherwise, $\ell(P) = \sum_{ e \in P \setminus E'} d(e) + \sum_{e
    \in {P \bigcap E'}} \min \curlyof{d(e), \optradius{d}} \geq \min
  \curlyof{d(s, p), \optradius{d}} \geq \min \curlyof{d(c(p), p),
    \optradius{d}}$. Therefore $d'(s, p) = \min_P \ell(p) \geq \min
  \curlyof{d(c(p), p), \optradius{d}}$.

\end{proof}

Consider the instance $\inst' = (V, d', k)$. Since, $\inst'$ is a
$2$-perturbed instance, the optimal clustering is given by
$\mathcal{C} = \curlyof{C_1, \ldots, C_k}$, and let $\optradius{d'}$
denote the cost of optimal solution. Using \reflemma{opt-unchanged-undirected}
we get, $\optradius{d'} = \optradius{d}$. Now, \refobs{kc-d-1} implies,
$\cost{d'}{B} = \max_{p \in V} d'(B, p) \leq \max_{p \in V} d'(c(p),
p) \leq \optradius{d} = \optradius{d'}$. Therefore, $B$ is a set of
optimal centers for $\inst'$. By perturbation resilience $B$ induces
the unique Voronoi partition $\mathcal{C}$ in $\inst'$. For
any $p \in V$, $d'(c(p), p)
\leq d'(p, B \setminus \setof{c(p)}) $ (follows from
\refobs{kc-d-1}). Since $B$ induces a unique Voronoi partition in $\inst'$, 
it must be the case $d'(c(p), p) < d'(p, B \setminus \setof{c(p)})$. This
also implies, that in clustering $\mathcal{C}$, every $p \in V$ belongs
to the same cluster as $c(p)$. Recall the definition of $c(p)$; $p$ was
assigned to $c(p)$ in the clustering induced by $B$ in the original
instance $\inst$. Thus the clusters in $\mathcal{C}'$ and
$\mathcal{C}$ are identical. This also proves that the Voronoi
clustering induced by $B$ in $\inst$ is unique. 
\end{proof}

\paragraph{Properties of
  \texorpdfstring{$2$}{2}-perturbation resilient
  \texorpdfstring{\kcenter}{k-center} instance:}
Angelidakis et al. \cite{AngelidakisMM17} showed that in the optimal
clustering of a $2$-perturbation resilient {\kcenter} instance, every
point is closer to its assigned center than to any point in a
different cluster. In fact they show this property for general
center based objectives, not just \kcenter. Here we observe that
\reftheorem{any-2} implies stronger structural
properties for \kcenter: (1) any point is closer to a point in its own cluster,
than to a point in a different cluster; (2) the distance between two
points in two different clusters is atleast the optimal radius (see
\reffigure{kc_1}). Rest of this section is devoted to proving
these properties.

\begin{figure}[!tbp]
  \centering
  \subfloat[Points in different optimal clusters are $\optradius{d}$ apart]{\includegraphics[width=0.48\textwidth]{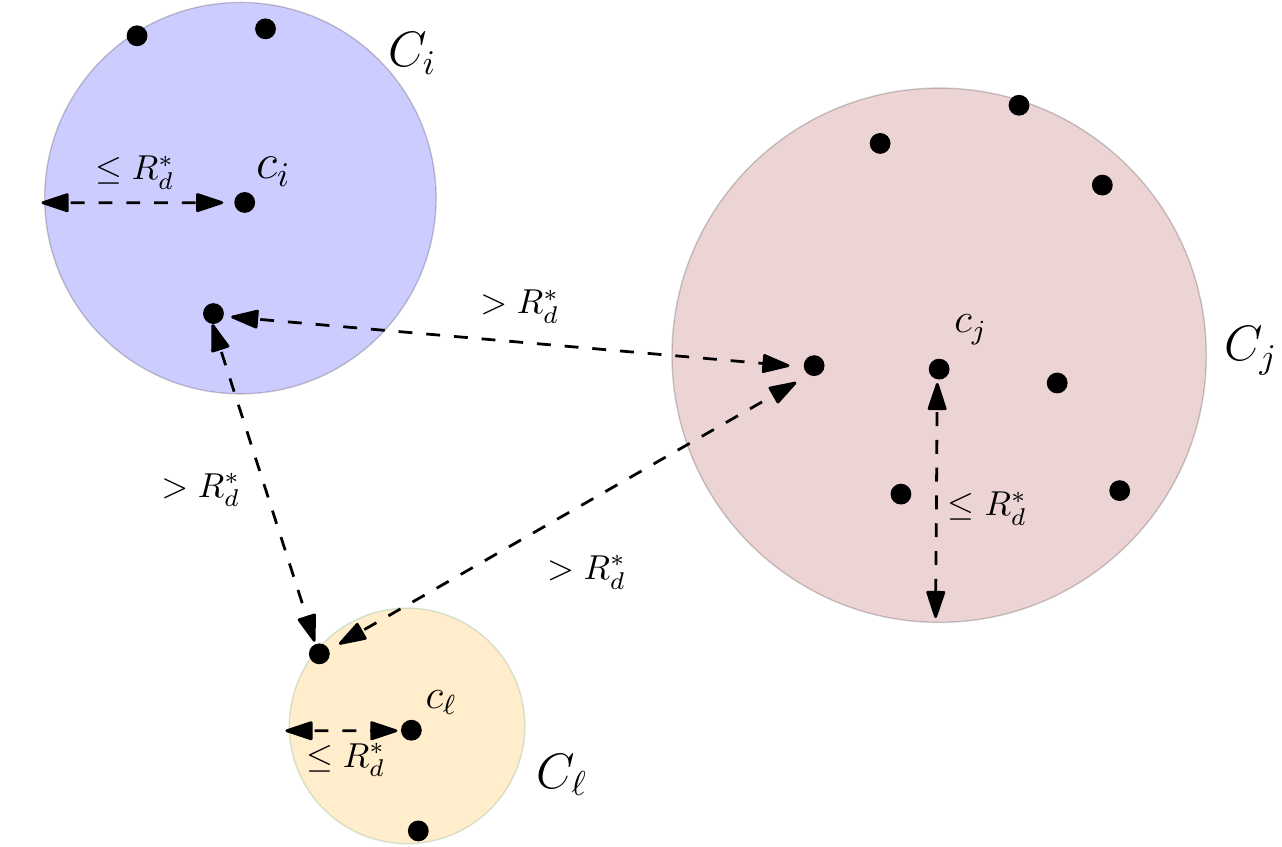}\labelfigure{kc_1}}
  \hfill
  \subfloat[Graph $G_R$ for $R < \optradius{d}$]{\includegraphics[width=0.48\textwidth]{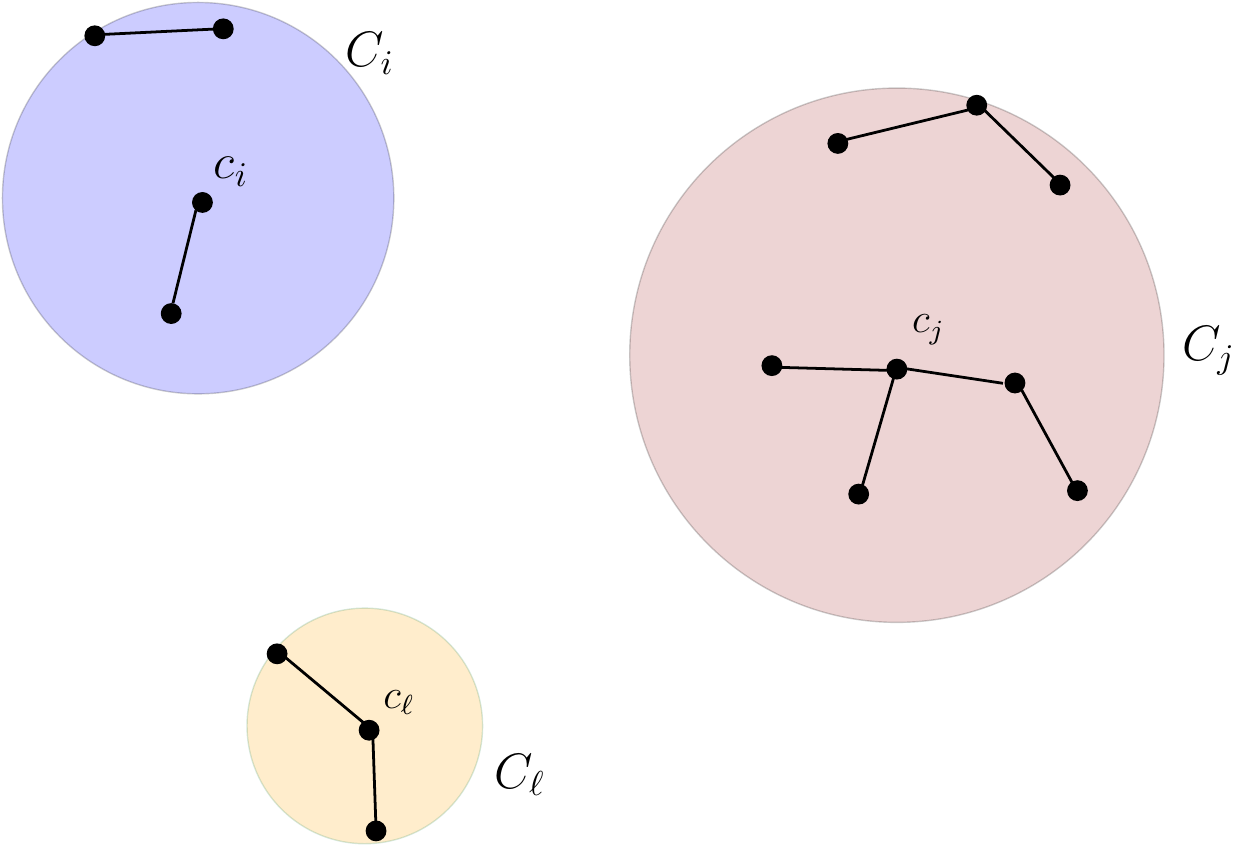}\labelfigure{kc_2}}
  \caption{Optimal Clusters in a $2$-perturbation resilient \kcenter instance}
\end{figure}

\begin{lemma}\labellemma{kc-sep-1}
Consider a $2$-perturbation resilient \kcenter instance $\inst = (V, d, k)$. 
Let $S = \curlyof{c_1, \ldots, c_k}$
be an optimal set of centers, and $\mathcal{C} = \curlyof{ C_1, \ldots C_k }$ 
be the corresponding unique optimal clustering. Consider any cluster $C_i$ with $\sizeof{C_i} \geq 2$, and let $p, w$ be any two points in $C_i$.
For any point $q$ in a different cluster $C_j$ ($i \neq j$), we have $d(p, q) > d(p, w)$.
\end{lemma}
\begin{proof}
  Note that the set of centers $B = S \setminus \setof{c_i, c_j}
  \Union \setof{w, q}$ gives a $2$-approximation. Therefore
  \reftheorem{any-2} immediately implies $d(w, p) < d(q, p)$.
\end{proof}

\begin{lemma}\labellemma{kc-sep}
  Consider a $2$-perturbation resilient \kcenter instance $\inst = (V,
  d, k)$.  Let $S = \curlyof{c_1, \ldots, c_k}$ be an optimal set of
  centers, and $\mathcal{C} = \curlyof{ C_1, \ldots C_k }$ be the
  corresponding unique optimal clustering. The optimal radius is
  $\optradius{d}$.  Consider any point $p \in V$, and let $p \in C_i$.
  For any point $q$ in a different cluster $C_j$ ($i \neq j$), we have
  $d(p, q) > \optradius{d}$.
\end{lemma}
\begin{proof}
  Assume for the sake of contradiction that the claim is not true,
  that is, there exists two points $p \in C_i$, and $q \in C_j$ such
  that $d(p, q) \leq \optradius{d}$. We claim both $C_i$ and $C_j$
  cannot have cardinality $1$; if it is the case then $k-1$ centers $S
  \setminus \{c_i \}$ will give solution of cost $\optradius{d}$, which would imply
  $\inst$ is not $2$-perturbation resilient (follows from
  \reflemma{voronoi-kc}). Assume without loss of generality that
  $\sizeof{C_i} > 1$. Now \reflemma{kc-sep-1}, coupled with our
  assumption $d(p, q) \leq \optradius{d}$ indicates, $\forall u \in
  C_i$, $d(u, p) < \optradius{d}$. First, note that this implies
  $\sizeof{C_j} > 1$, as otherwise the $k-1$ centers $S \setminus
  \{c_i, c_j \} \Union \setof{p}$ will give an optimal solution, which
  cannot happen for a $2$-PR instance. Second, by triangle inequality,
  we have, $\forall u \in C_i$, $d(q, u) < 2 \cdot
  \optradius{d}$. Let $q'$ be any arbitrary point in $C_j
  \setminus \{ q \}$. The set of centers $B = S \setminus \setof{c_i,
    c_j} \Union \setof{q, q'}$ gives a $2$-approximation since every
  point in $C_i \bigcup C_j$ is within $2\optradius{d}$ of
  $\{q,q'\}$. However, the Voronoi partition induced by $B$ is clearly
  different from $\mathcal{C}$, as $C_i$ is no longer a cluster. This
  contradicts \reftheorem{any-2}.
\end{proof}

\subsection{LP Integrality}
Now, we show that as a consequence of \reflemma{kc-sep}, the LP
relaxation for \kcenter is integral. Given an instance $\inst = (V, d,
k)$ of {\kcenter} and a parameter
$R \geq 0$, we define the graph (also called \emph{threshold} graph)
$G_R = (V, E_R)$, where $E_R = \{ (u, v) : u, v \in V, d(u, v) \leq
R\}$. For a vertex $v$, let $\nbr{v} = \{ u: (u, v) \in E_R\} \bigcup \{
v \}$ be the neighbors (including itself).  Observe, for any $R \geq
\optradius{d}$, where $\optradius{d}$ is the optimal solution cost of
$\inst$, there exists a set of $k$ centers $S \subseteq V$, such that
$S$ \textit{covers} $V $ in $G_R$, i.e. $\bigcup_{c \in S} \nbr{c} = V
$. Given a parameter $R$, we can define the following LP on graph
$G_R$. We use $y_v$ as an indicator variable for open centers, and
$x_{uv}$ to denote if $v$ is assigned to $u$.
\vspace{-.5em}\begin{center}
  \fbox{\begin{minipage}[t]{0.5\textwidth}\vspace{-.4cm}
\begin{align}\labelequation{kc-LP}
& \sum_{u \in V} y_u \leq k \quad & \tag{\textbf{kc-LP}}\nonumber\\
& x_{uv} \leq y_u \quad &\forall v \in V, u \in V  \nonumber \\
& \sum_{u \in \nbr{v}} x_{uv} \geq 1 \quad &\forall v \in V \nonumber \\
& y_v, x_{uv} \geq 0 \quad & \nonumber
\end{align}
\end{minipage}}
\end{center}

The minimum $R$ for which \refequation{kc-LP} is feasible provides a
lower bound on the optimum solution, and is the standard relaxation
for \kcenter. It easy to see for all $R \geq \optradius{d}$
\refequation{kc-LP} is feasible. Further, it is well-known that the
integrality gap is $2$, that is, for all $R < \optradius{d}/2$, the LP
is infeasible.  However, if the \kcenter instance is $2$-perturbation
resilient, we can show that LP has no integrality gap.

\begin{theorem}\labeltheorem{kc-lp}
  Consider a $2$-perturbation resilient instance $\inst = (V, d, k)$
  of {\kcenter}.  Let $\optradius{d}$ be the cost of the optimal
  solution. Then, for any $R < \optradius{d}$, \refequation{kc-LP} is
  infeasible.
\end{theorem}

\begin{proof}
  Let $C_1, \ldots C_k$ be the unique optimal clustering of instance
  $\inst = (V, d, k)$, with optimal radius $\optradius{d}$. Consider
  an arbitrary $R < \optradius{d}$, and let $G_R$ denote the
  corresponding threshold graph. Recall, in graph $G_R$ the vertex set
  is $V$, and the edge set $E_R = \curlyof{(u, v) : d(u, v) \leq
    R}$. According to \reflemma{kc-sep}, in a $\pr{2}$ instance, two
  points belonging to two different optimal clusters are separated by
  a distance of strictly more than $\optradius{d}$. Since $\inst$ is
  $\pr{2}$, graph $G_R$ has a simple structure --- for any $v \in C_i, i
  \in [k]$, $\nbr{v} \subseteq C_i$. Or in other words, the connected
  components of $G_R$ are subsets of the optimal clusters (see
  \reffigure{kc_2}).

  Suppose, the \kcenter LP (\refequation{kc-LP}) defined over graph
  $G_R$ is feasible, and $(x,y)$ is the feasible fractional
  solution. Since every point is fully covered, and it can be covered
  only by its neighbors in $G_R$, we have, for all $C_i$, $\sum_{u \in
    C_i} y_u \geq 1$. Since, $\sum_{v \in V} y_v \leq k$, and the
  clusters $C_1, \ldots, C_k$ are disjoint, we have $\sum_{u \in C_i}
  y_u = 1$, for each $i$.

  From the definition of $\optradius{d}$, there is an optimum cluster
  $C_t$ such that $\min_{c \in C_t}\max_{v \in C_t} d(c, v) =
  \optradius{d}$.  Let $C_t' = \curlyof{u \in C_t : y_u > 0}$. As we
  argued earlier, $\sum_{u \in C_t} y_u = \sum_{u \in C_t'} y_u =
  1$. Further, since for every $v \in C_t$, $\nbr{v} \subseteq C_t$,
  and $v$ needs to be covered, we must have $C_t' \subseteq
  \nbr{v}$. Consider any $c \in C_t'$.  Note that for every $v \in
  C_t$, $c$ is a neighbor of $v$ in graph $G_R$, i.e. $d(c, v) \leq R
  < \optradius{d}$. This implies, $\max_{v \in C_t} d(c, v) <
  \optradius{d}$ which is a contradiction.
\end{proof}

\section{LP Integrality of {\asymkcenter} under Perturbation Resilience}\labelsection{asymkc}

We start with an LP relaxation for \asymkcenter problem by considering
an unweighted directed graph on node set $V$. Specifically, for a
parameter $R \geq 0$, we define the directed graph $G_R = (V, E_R)$,
where $E_R = \{ (u, v) : u, v \in V, d(u, v) \leq R\}$. For a node
$v$, let $\innbr{v} = \{ u: (u, v) \in E_R\} \bigcup \setof{v}$ denote
the in-neighbors, and $\outnbr{v} = \{ u: (v, u) \in E_R\} \bigcup
\setof{v}$ be the out-neighbors (including itself). Observe, for any
$R \geq \optradius{d}$, there exists a set of $k$ centers $S \subseteq
V$, such that $S$ \textit{covers} $V$ in $G_R$, i.e. $\bigcup_{c \in
  S} \outnbr{c} = V$. Thus, given a parameter $R$, we can define the
following LP relaxation on graph $G_R$. We use $y_v$ as an indicator
variable for open centers, and $x_{uv}$ to denote if $v$ is assigned
to $u$.

\vspace{-.5em}\begin{center}
\fbox{\begin{minipage}[t]{0.5\textwidth}\vspace{-.4cm}
\begin{align}\labelequation{asym-kc-LP}
& \sum_{u \in V} y_u \leq k \quad & \tag{\textbf{asym-kc-LP}}\nonumber\\
& x_{uv} \leq y_u \quad &\forall v \in V, u \in V  \nonumber \\
& \sum_{u \in \innbr{v}} x_{uv} \geq 1 \quad &\forall v \in V \nonumber \\
& y_v, x_{uv} \geq 0 \quad & \nonumber
\end{align}
\end{minipage}}
\end{center}
For \asymkcenter, Archer \cite{Archer01} showed that the integrality
gap is atmost $O(\log ^* k)$, Infact it is tight within a constant factor
\cite{ChuzhoyGHKKKN05}. The main result of the section is captured
by the following theorem.

\begin{theorem}\labeltheorem{asymkc-lp}
  Let $\inst = (V, d, k)$ be a $2$-perturbation resilient instance of
  {\asymkcenter} and let $\optradius{d}$ be the cost of the optimal
  solution. Then, for any $R < \optradius{d}$,
  \refequation{asym-kc-LP} is infeasible.
\end{theorem}

\subsection{Properties of \texorpdfstring{$2$}{2}-perturbation resilient \texorpdfstring{\asymkcenter}{Asymmetric k-center} instance}
In \refsection{kc} we showed that the clusters in an optimal solution to a
$\pr{2}$ \kcenter instance have a strong separation property: $d(p,q)
> \optradius{d}$ if $p,q$ are in different clusters.  For \asymkcenter
the asymmetry in the distances does not permit a such a strong and
simple separation property.  However, we can show slightly weaker
properties: (1) every optimal center is separated from any point in a
different cluster by at least $\optradius{d}$; (2) points in a cluster
which are far off from \emph{core} points (these points have small
distance "to" correponding cluster centers) in the cluster, are
well-separated from core points of other clusters as well (See
\reffigure{asymkc_1}, \reffigure{asymkc_2}). These properties suffice
to prove our desired theorem. The rest of the section is dedicated to
proving these properties.

\begin{figure}[!tbp]
  \centering
  \subfloat[Cluster centers are atleast $\optradius{d}$ away from points in different cluster]{\includegraphics[width=0.48\textwidth]{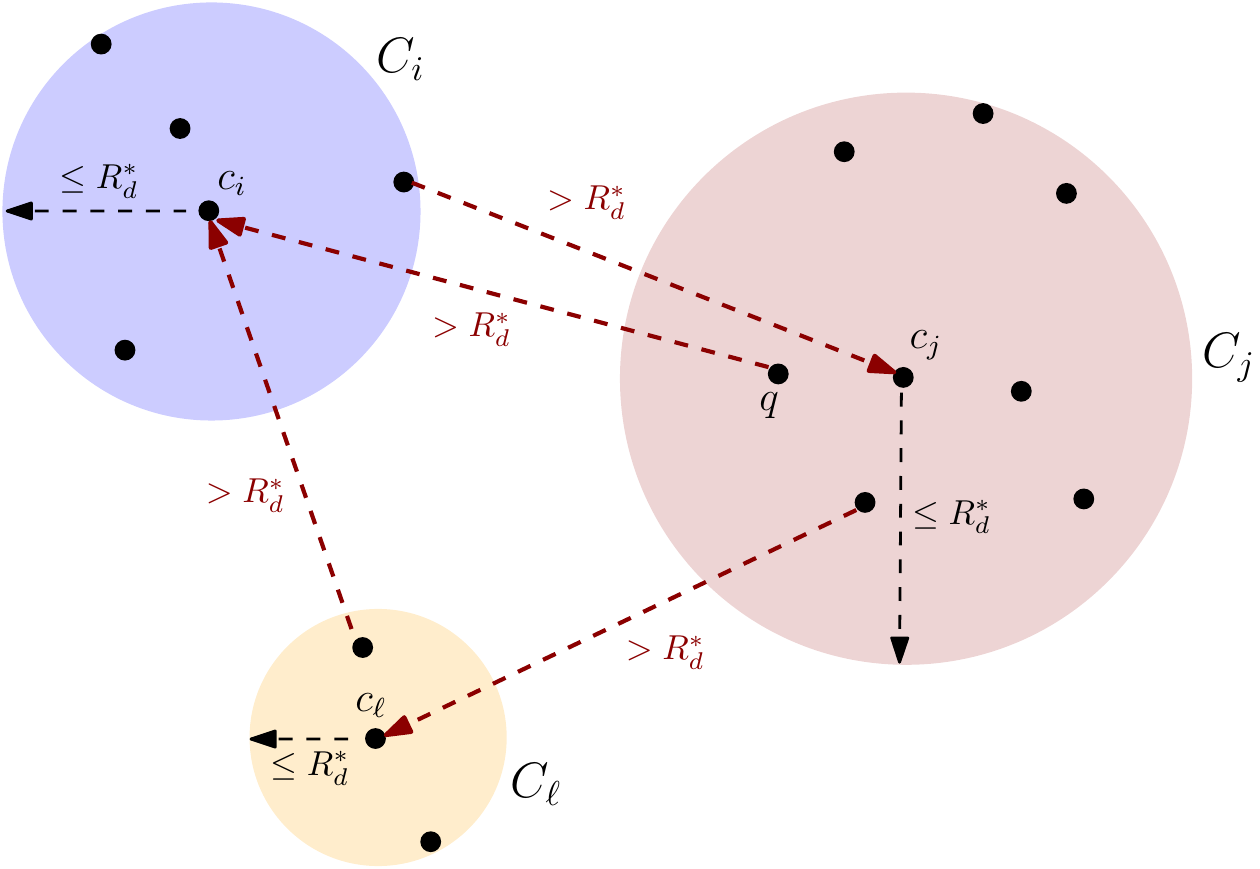}\labelfigure{asymkc_1}}
  \hfill
  \subfloat[Cluster points $\optradius{d}$ away from cluster core points, are $\optradius{d}$ away from different cluster core points]{\includegraphics[width=0.48\textwidth]{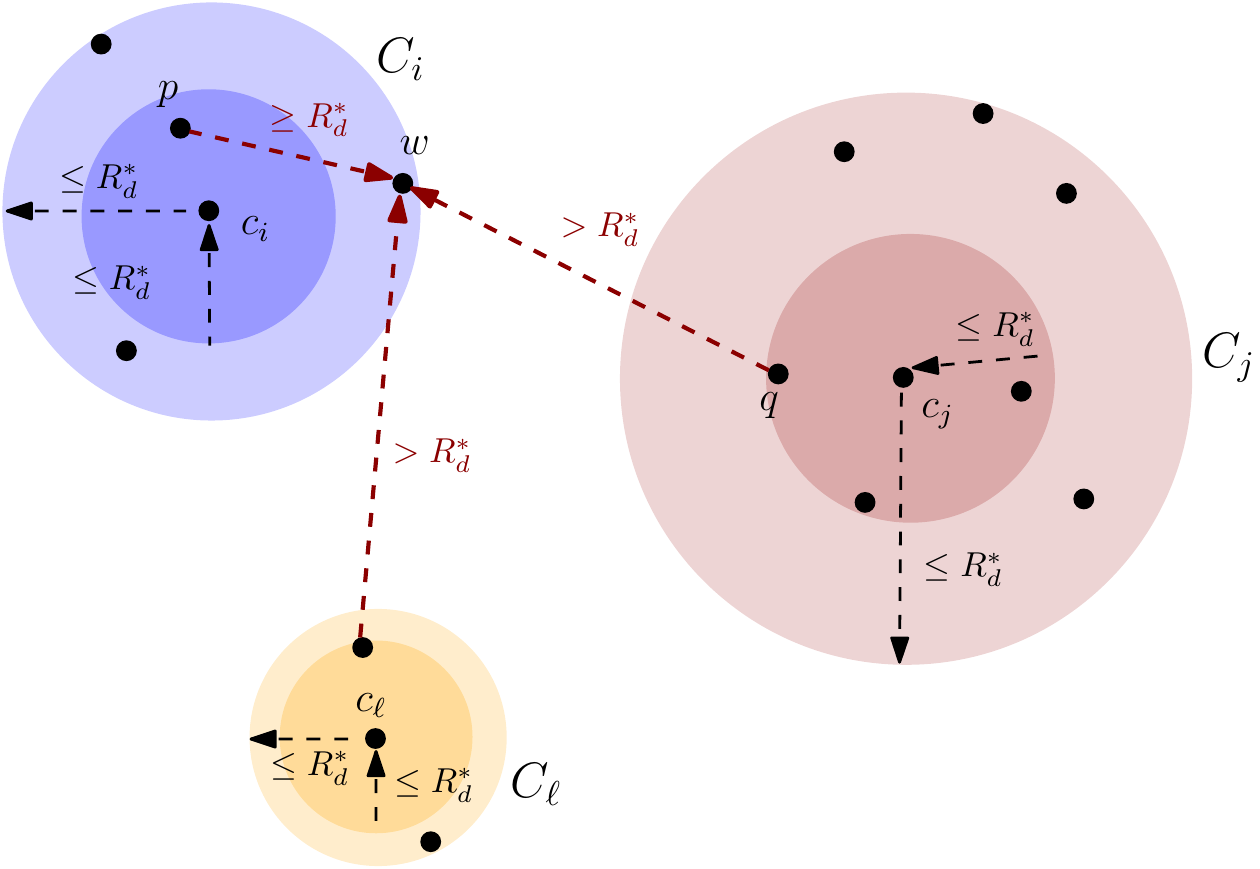}\labelfigure{asymkc_2}}
  \caption{Properties of a $2$-perturbation resilient \asymkcenter instance}
\end{figure}

\begin{lemma}\labellemma{asymkc-sep-1}
  Consider a $2$-perturbation resilient \asymkcenter instance $\inst =
  (V, d, k)$.  Let $\mathcal{C} = \curlyof{ C_1, \ldots C_k }$ be the
  unique optimal clustering, induced by a set of centers $S =
  \curlyof{c_1, \ldots, c_k}$. Let the optimal radius be
  $\optradius{d}$.  Consider any center $c_i$. Then for any point $q$
  in a different cluster $C_j$ ($i \neq j$), we have $d(q, c_i) >
  \optradius{d}$.
\end{lemma}
\begin{proof}
  Assume for the sake of contradiction, that the claim is false. That is, there exists a point $q \in C_j$, such that $d(q,
  c_i) \leq \optradius{d}$. We
  construct a distance function $d'$, which is a metric
  $2$-perturbation of $d$. And show that in the instance thus
  constructed, the optimal clustering is not unique, which contradicts
  the definition of perturbation resilience.

  We define $d'$ as follows: consider the complete directed graph $G$
  on vertices $V$.  Let $E' = \curlyof{(q, v): v \in C_i }$.  The edge
  lengths in graph $G$ are given by the function $\ell$, where for any
  edge $(u, v)$,
\[
\ell(u, v) = 
  \begin{cases} 
    \min \curlyof{d(u, v), \optradius{d} } \quad& (u,v) \in E' \\
    d(u, v) \quad& \text{otherwise}
  \end{cases}
\]
For any pair of points $u, v$, the distance $d'(u, v)$ is the shortest
path distance between $u$ and $v$ in graph $G$, using $\ell$.
\begin{obs}\labelobs{asymkc-1-perturbed}
$d'$ is a metric $2$-perturbation of $d$.
\end{obs}
\begin{proof}
For any $v \in C_i$, by triangle inequality, $d(q, v) \leq d(q, c_i) + d(c_i, v) \leq 2 \cdot \optradius{d}$. 
Therefore, for any $(u, v) \in E'$,
$\ell(u, v) = \min \curlyof{d(u, v), \optradius{d} } \geq \frac{d(u, v)}{2}$. For any other $(u, v)$,
by definition $\ell(u, v) = d(u, v)$. That is, for any edge $(u, v)$ we have,
$\frac{d(u, v)}{2} \leq \ell(u,v) \leq d(u, v)$. The shortest path distance function $d'$, defined on such a graph $G$, can be easily shown is 
a metric $2$-perturbation of $d$ (details are chalked out in proof \reflemma{d-perturb}).
\end{proof}

Consider the instance $\inst' = (V, d', k)$. Since, $\inst'$ is a
$2$-perturbed instance, the optimal clustering is given by
$\mathcal{C} = \curlyof{C_1, \ldots, C_k}$. Let $S' = \curlyof{c_1',
  \ldots, c_k'}$ be the optimal set of centers. Further,
$\optradius{d'}$ denotes the cost of optimal solution. And we can
show, $\optradius{d'} = \optradius{d}$ (follows from \reflemma{opt-unchanged}).

\begin{description}[style=unboxed,leftmargin=0cm]
\item [Case 1: $q \neq c_j'$.] 
Consider the set of  
centers $S'' = S' \setminus \setof{c_i'} \bigcup \setof{q}$.
Let $\mathcal{C''}$ be a Voronoi partition induced by $S''$. 
For any point $u \in C_{\ell}$, where $\ell \neq i$, 
$d'(S'', u) \leq d'(c_\ell', u) \leq \optradius{d'}$. For any point $u \in C_i$,
 $d'(S'', u) \leq d'(q, u) \leq \optradius{d} = \optradius{d'}$. 
 That is, for any point $u \in V$,
$d'(S'', u) \leq \optradius{d'}$. Therefore $\cost{d'}{\mathcal{C''}, S''} \leq \optradius{d'}$.
Now, in clustering $\mathcal{C''}$, the points $q$ and $c_j'$ are in different clusters,
which is not true for $\mathcal{C}$. Thus the optimal clustering is not unique, and this leads 
to contradiction.

\item [Case 2: $q = c_j'$.]  Consider the set of $k-1$ centers $S'' =
  S' \setminus \setof{c_i'}$. Let $\mathcal{C''}$ be a Voronoi
  partition induced by $S''$.  As in the previous case, we can show
  $d(S'', u) \leq \optradius{d'}$, for any $u \in V$, implying
  $\cost{d'}{\mathcal{C''}, S''} \leq \optradius{d'}$. Therefore, we
  have a $k-1$ clustering of $\inst'$ with cost at most the optimal
  cost.  Then, by \reflemma{voronoi-kc}, the optimal clustering of
  $\inst'$ is not unique. This contradicts the definition of
  perturbation resilience.
\end{description}
\end{proof}

The next lemma formalizes the notion of core points and the property
they enjoy.

\begin{lemma}\labellemma{asymkc-sep-2}
  Consider a $2$-perturbation resilient \asymkcenter instance $\inst =
  (V, d, k)$.  Let $\mathcal{C} = \curlyof{ C_1, \ldots C_k }$ be the
  unique optimal clustering induced by a set of centers $S =
  \curlyof{c_1, \ldots, c_k}$. Let the optimal radius is
  $\optradius{d}$. Suppose $p \in C_i$ and $q \in C_j$ where $i \neq
  j$ and $d(p,c_i) \le \optradius{d}$ and $d(q,c_j) \le
  \optradius{d}$. Then for any $w \in C_i$ such that $d(p,w) \ge
  \optradius{d}$ we have $d(q,w) > \optradius{d}$.  
\end{lemma}

\begin{proof}
  Consider a triplet of points $p, w \in C_i$ and $q \in C_j$, where
  $d(p, c_i), d(q, c_j) \leq \optradius{d}$, and $d(p, w) \geq
  \optradius{d}$. Assume for the sake of contradiction, $d(q, w) \leq
  \optradius{d}$.  We construct a distance function $d'$, which is a
  metric $2$-perturbation of $d$. Next we show that in the \asymkcenter instance constructed using $d'$, the optimal clustering is not unique, which contradicts
  the definition of perturbation resilience.

  We define $d'$ as follows: consider the complete directed graph $G$
  on vertices $V$.  Let $E' = \curlyof{(p, v): v \in C_i } \bigcup
  \curlyof{(q, v): v \in C_j }$.  The edge lengths in graph $G$ are
  given by the function $\ell$, where for any edge $(u, v)$,
\[
\ell(u, v) = 
  \begin{cases} 
    \min \curlyof{d(u, v), \optradius{d} } \quad& (u,v) \in E' \\
    d(u, v) \quad& \text{otherwise}
  \end{cases}
\]
For any pair of points $u, v$, the distance $d'(u, v)$ is the shortest path distance between
$u$ and $v$ in graph $G$, using $\ell$. 
\begin{obs}\labelobs{asymkc-2-perturbed}
$d'$ is a metric $2$-perturbation of $d$.
\end{obs}
\begin{proof}
For any $v \in C_i$, by triangle inequality, $d(p, v) \leq d(p, c_i) + d(c_i, v) \leq 2 \cdot \optradius{d}$. 
Similarly, for any $v \in C_j$, by triangle inequality, $d(q, v) \leq 2 \cdot \optradius{d}$
Therefore, for any $(u, v) \in E'$,
$\ell(u, v) = \min \curlyof{d(u, v), \optradius{d} } \geq \frac{d(u, v)}{2}$. For any other $(u, v)$,
by definition $\ell(u, v) = d(u, v)$. As we previously stated (in \reflemma{d-perturb}), $d'$ defined on
graph $G$, satisfying the property $\frac{d(u, v)}{2} \leq \ell(u, v) \leq d(u,v)$ is a metric $2$-perturbation of $d$.
\end{proof}

\begin{obs}\labelobs{asymkc-3-perturbed}
For any point $v \in C_i$, we have $d'(p, v) \leq \optradius{d}$. Similarly for any $v \in C_j$,
$d'(q, v) \leq \optradius{d}$. In particular, for point $w$, $d'(p, w) = \optradius{d}$ and
$d'(q, w) \leq \optradius{d}$.
\end{obs}
\begin{proof}
For any point $v \in C_i$, by definition, $\ell(p, v) \leq \optradius{d}$. Since, $d'(p, v) \leq \ell(p, v)$
the claim follows. Similarly, for any $v \in C_j$, it is easy to see $d'(q, v) \leq \optradius{d}$.

Now, consider point $w$. Let $P$ be any directed path from $p \leadsto w$ in graph $G$, excluding 
the single edge path $(p, w)$. If $P$ includes an edge from $E'$, by triangle inequality,
$\ell(P) = \sum_{e \in P} \ell(e) = \sum_{e \in P} d(e) \geq d(p, w) \geq \optradius{d}$.
The last inequality follows from our choice of $p, w$ at the outset. Otherwise,
if $P$ includes atleast one edge from $E'$,
$\ell(P') = \sum_{e \in P' \bigcap E'} \ell(e) + \sum_{e \in P' \setminus E'} \ell(e) \geq 
\sum_{e \in P' \bigcap E'} \min \curlyof{d(e), \optradius{d}} + \sum_{e \in P' \setminus E'} d(e) \geq \min \curlyof{d(p, w), \optradius{d}}  \geq \optradius{d}$. Finally, 
by our definition $\ell(p, w) = \optradius{d}$. Therefore, $d'(p, w) = \optradius{d}$.
The final observation, $d'(q, w) \leq \optradius{d}$, follows from our 
assumption $d(q, w) \leq \optradius{d}$.
\end{proof}

Consider the instance $\inst' = (V, d', k)$. Since, $\inst'$ is a $2$-perturbed instance, the optimal
clustering is given by $\mathcal{C} = \curlyof{C_1, \ldots, C_k}$. Let $S' = \curlyof{c_1', \ldots, c_k'}$ be 
the optimal set of centers. Further, $\optradius{d'}$ denotes the cost of optimal solution. We can show, $\optradius{d'} = \optradius{d}$.

Consider the set of  
centers $S'' = \parof{S' \setminus \curlyof{c_i', c_j'}} \bigcup \curlyof{p, q}$.
Let $\mathcal{C''}$ be a Voronoi partition induced by $S''$. 
For any point $u \in C_{\ell}$, where $\ell \neq i, j$, 
$d'(S'', u) \leq d'(c_\ell', u) \leq \optradius{d'}$. For any point $u \in C_i$,
 $d'(S'', u) \leq d'(p, u) \leq \optradius{d} = \optradius{d'}$. Similarly, 
 for any point $u \in C_j$,
 $d'(S'', u) \leq d'(q, u) \leq \optradius{d} = \optradius{d'}$
 That is, for any point $u \in V$,
$d'(S'', u) \leq \optradius{d'}$. Therefore $\cost{d'}{\mathcal{C''}, S''} \leq \optradius{d'}$.
Recall by \refobs{asymkc-3-perturbed}, $d'(p, w) = \optradius{d}$, while $d'(q, w) \leq \optradius{d}$.
Therefore, we can assume without loss of generality, in $\mathcal{C''}$, $w$ and $p$ are not in same 
cluster. This however is not true for $\mathcal{C}$, implying $\mathcal{C''} \neq \mathcal{C}$. 
 Thus the optimal clustering of $\inst'$ is not unique, and this contradicts the definition
 of perturbation resilience.
\end{proof}

\subsection{Proof of \reftheorem{asymkc-lp}}
Let $C_1, \ldots C_k$ be the unique optimal clustering of instance
$\inst = (V, d, k)$, with optimal radius $\optradius{d}$. Consider an
arbitrary $R < \optradius{d}$, and let $G_R$ denote the corresponding
threshold graph. Recall, graph $G_R$ is a directed graph defined over
vertex set $V$, and the edge set $E_R = \curlyof{(u, v) : d(u, v) \leq
  R}$. Suppose, the \asymkcenter LP (\refequation{asym-kc-LP}) defined
over graph $G_R$ is feasible, and $(x,y)$ is a feasible fractional
solution.

From \reflemma{asymkc-sep-1}, in a $\pr{2}$ instance, 
we have the following: if $q \not \in C_i$ then $d(q,c_i) > \optradius{d} > R$.
This implies that, in the graph $G_R$, for any $c_i, i \in [k]$,
$\innbr{c_i} \subseteq C_i$.  Let $C_i' = \curlyof{u \in \innbr{c_i}:
  y_u > 0}$.  Since $(x, y)$ is a feasible solution, we must have
$\sum_{u \in C_i'} y_u \geq 1$.  Since, $\sum_{v
  \in V} y_v \leq k$, and the clusters $C_1, \ldots, C_k$ are
disjoint, we have $\sum_{u \in C_i} y_u = \sum_{u \in C_i'} y_u = 1$,
for all $i \in [k]$.

From the definition of $\optradius{d}$ there must be a cluster $C_t$
such that $\min_{c \in C_t}\max_{v \in C_t} d(c, v) = \optradius{d}$.
Consider its center $c_t$ and let $p \in C_t'$. Clearly $d(p, c_t)
\leq R < \optradius{d}$.  Furthermore, since $C_t$ is the largest
radius cluster, there exists $w \in C_t$, such that $d(p, w) \geq
\optradius{d}$. Therefore in graph $G_R$, $p \notin \innbr{w}$.  For
any other cluster $C_j$, by \reflemma{asymkc-sep-2}, for any point $q
\in C_j'$, we have $d(q, w) > \optradius{d}$. That is, $\innbr{w}
\bigcap C_j' = \emptyset$, for any $j \neq t$. This implies $w$ can be
covered only by points that belong to $C_t'$.  Therefore $\sum_{u \in
  \innbr{w}} x_{uw} \leq \sum_{u \in C_t' - p} y_u <1$ since $y_p >
0$. This contradicts feasibility of $(x, y)$.

\section{LP Integrality of {\kcenterout} under Perturbation Resilience}\labelsection{kco}
In this section we now consider the {\kcenterout} problem.  Recall
that an instance $\inst = (V, d, k, z)$ consists of a finite metric
space $(V,d)$ an integer $k$ specifying the number of centers and 
an integer $z < |V|$ specifying the number of outliers that are allowed.
One can write a natural LP relaxation for this problem as follows.
As before, for a parameter $R \geq 0$, we define the
graph $G_R = (V, E_R)$, where $E_R = \{ (u, v) : u, v \in V, d(u, v)
\leq R\}$. For a node $v$, let $\nbr{v} = \{ u: (u, v) \in E_R\} \cup
\{ v \}$ be the neighbors (including itself).  Observe, for any $R
\geq \optradius{d}$, there exists a set of $k$ centers $S \subseteq
V$, and a set of outliers $Z$ with $\sizeof{Z} \leq z$, such that $S$
\textit{covers} $V \setminus Z$ in $G_R$, i.e. $\cup_{c \in S} \nbr{c}
= V \setminus Z$. Thus, given a parameter $R$, we can define the
following LP relaxation on graph $G_R$. We use $y_v$ as an indicator
variable for open centers, and $x_{uv}$ to denote if $v$ is assigned
to $u$.  \vspace{-.5em}\begin{center}
  \fbox{\begin{minipage}[t]{0.5\textwidth}\vspace{-.4cm}
\begin{align}\labelequation{kco-LP}
& \sum_{u \in V} y_u \leq k \quad & \tag{\textbf{kco-LP}}\nonumber\\
& x_{uv} \leq y_u \quad &\forall v \in V, u \in V  \nonumber \\
& \sum_{u \in V} x_{uv} \leq 1 \quad &\forall v \in V \nonumber \\
& \sum_{v \in V} \sum_{u \in V } x_{uv} \geq n - z \nonumber \\
& x_{uv} = 0 \quad & \forall v \in V, u \notin \nbr{v} \nonumber \\
& y_v, x_{uv} \geq 0 \quad & \nonumber
\end{align}
\end{minipage}}
\end{center}
The \refequation{kco-LP} is feasible for all $R \geq \optradius{d}$. The main theorem we prove in this 
section is as follows:

\begin{theorem}\labeltheorem{kco-lp-integrality}
Given a $2$-perturbation resilient instance $\inst = (V, d, k, z)$ of \kcenterout with optimal cost $\optradius{d}$, 
\refequation{kco-LP} is infeasible for any $R < \optradius{d}$.
\end{theorem}

\subsection{Properties of $2$-perturbation resilient {\kcenterout} instance}
For \kcenterout we extend the properties from \refsection{kc} that
hold for $2$-perturbation resilient instances.  The first property
shows that if $p$ is a non-outlier point $p$ and $q$ is any point not
in the same cluster as $p$ ($q$ could be an outlier) then $d(p,q) >
\optradius{d}$. The second property is that for any outlier point $q$,
the number of outliers in a ball of radius $2\optradius{d}$ is
small. Specifically the number of points is strictly smaller than the
size of the smallest cluster in the optimum clustering. This property
makes intuitive sense, for otherwise $q$ can define another cluster
with outlier points and contradict the uniqueness of the clustering in
after perturbation. We formally state them below after setting up the
required notation.

\begin{figure}[!tbp]
  \centering
  \subfloat[Points in an optimal cluster separated by $\optradius{d}$ from points outside that cluster]{\includegraphics[width=0.48\textwidth]{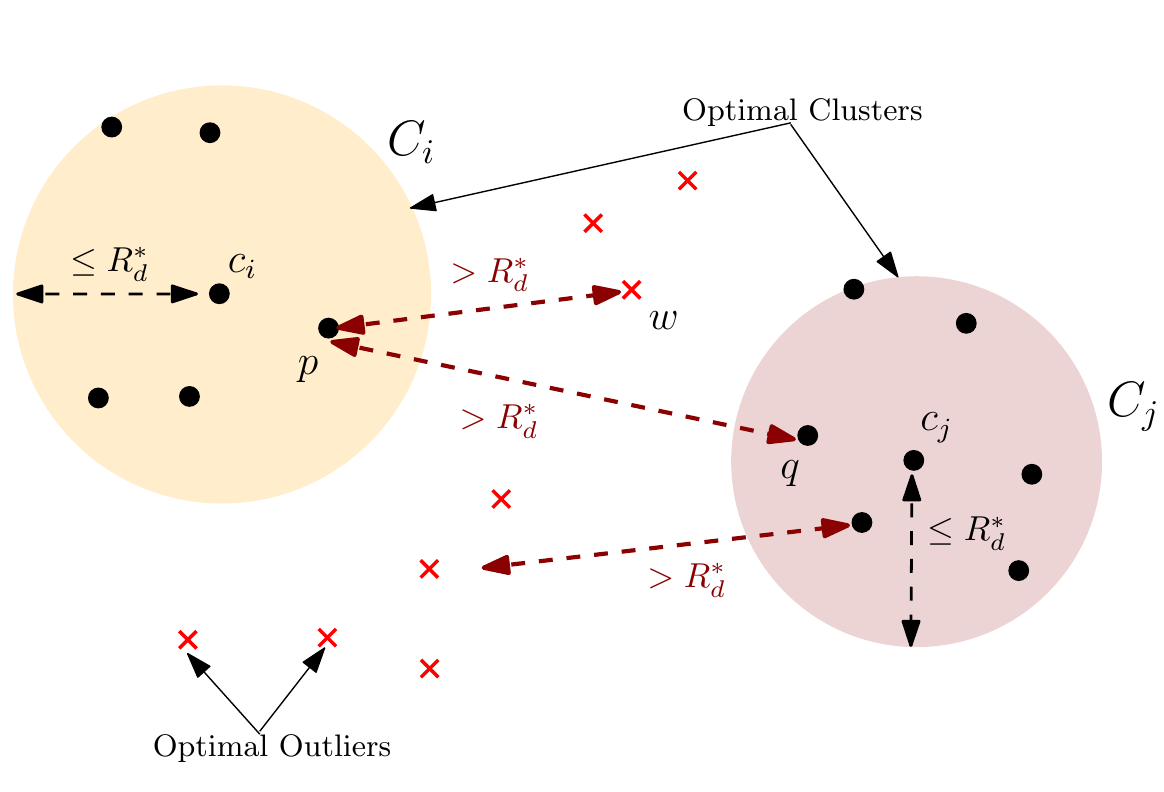}\labelfigure{kco_1}}
  \hfill
  \subfloat[Sparse Neighborhood of an outlier]{\includegraphics[width=0.48\textwidth]{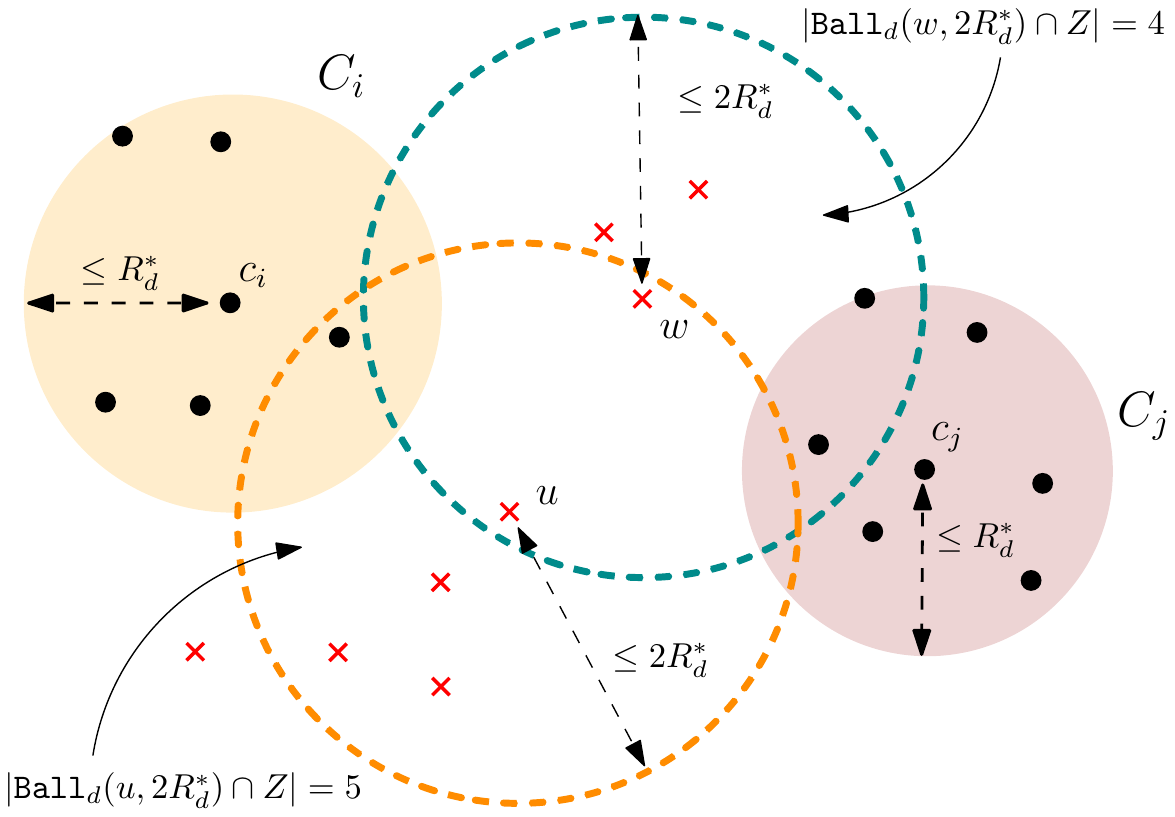}\labelfigure{kco_2}}
  \caption{Properties of a $2$-perturbation resilient \kcenterout instance}
\end{figure}

Let $\inst = (V, d, k, z)$ be a $2$-outlier perturbation resilient
{\kcenterout} instance. Let $\mathcal{C} = \setof{C_1, \ldots, C_k}$
be the optimum clustering, and $Z$ be the set of outliers in the optimal
solution of $\inst$. Further, let $S = \setof{c_1, \ldots, c_k}$ be
the optimal centers inducing the clustering $\mathcal{C}$. Let the optimal
cost be $\optradius{d}$. For each optimal cluster $C_i$, $n_i =
\sizeof{C_i}$ denotes its cardinality.  Additionally, given a
point $u \in Z$, and radius $R$, let $\ball{d}{u, R} = \{ v \in V : d(u, v)
\leq R\}$ be the set of points in a ball of radius $R$ centered at
$u$.


The two main structural properties of an $\opr{2}$ \kcenterout
instance we show are as follows (See \reffigure{kco_1},
\reffigure{kco_2}):

\begin{lemma}\labellemma{kco-sep}
  Consider any non-outlier point $p \in V \setminus Z$, and let $p \in C_i$. For
  all $q \notin C_i$, $d(p, q) > \optradius{d}$.
\end{lemma}

\begin{lemma}\labellemma{kco-out-ball}
  For any outlier $p \in Z$, we have $\sizeof{\ball{d}{p, 2 \cdot
      \optradius{d}} \Inter Z} < \min \curlyof{n_1, \ldots, n_k}$.
\end{lemma}

We observe that a much weaker version of the preceding lemma suffices
for our proof of LP integrality. The weaker version states that
$\sizeof{\ball{d}{p, \optradius{d}} \Inter Z} < \min \curlyof{n_1,
  \ldots, n_k}$. For if the statement is false, we could replace the
smallest cluster with the cluster $\ball{d}{p, \optradius{d}}$;
this gives an alternate clustering with at most $z$ outliers and the
same optimum radius contradicting the uniqueness of the optimum solution.

We now prove the two lemmas.

\subsubsection{Proof of \reflemma{kco-sep}} 
We prove \reflemma{kco-sep} by splitting it into two cases. We first
show that the the lemma holds true for all $q \in Z$. Next we show
that the lemma holds true, even when $q \in C_j$ ($j \neq i)$.

\begin{lemma}\labellemma{kco-q-outlier}
  Consider any point $p \in V \setminus Z$, and let $p \in C_i$. For
  all $q \in Z $, $d(p, q) > \optradius{d}$.
\end{lemma}
\begin{proof}
  Assume that the claim is not true, that is, there exists $q \in Z$
  such that $d(p, q) \leq \optradius{d}$.  Since $p \in C_i$, we have
  $d(c_i, p) \leq \optradius{d}$. Therefore by triangle inequality,
  $d(c_i, q) \leq 2 \cdot \optradius{d}$.

  We now define a metric distance function $d'$, which is
  $2$-perturbation of $d$. To this end, consider the complete
  undirected graph $G$ on vertex set $V$. The edge lengths in graph
  $G$ are given by the function $\ell$, where for any edge $(u, v)$,
\[
\ell(u, v) = 
  \begin{cases} 
    \min \curlyof{d(u, v), \optradius{d}}  \quad& u = c_i, v = q \\
    d(u, v) \quad& \text{otherwise}
  \end{cases}
\]
For any pair of points $u, v$, the distance $d'(u, v)$ is the shortest
path distance between $u$ and $v$ in graph $G$, using $\ell$. The
following observation is easy to see since $d(u,v)/2 \le \ell(u,v) \le
d(u,v)$ for every pair $(u,v)$ (follows from \reflemma{d-perturb-undirected}).

\begin{obs}\labelobs{kcenter-obs-2}
$d'$ is a metric $2$-perturbation of $d$. 
\end{obs}

Consider the instance $\inst' = (V, d', k, z)$. Since, $\inst'$ is a
$2$-perturbed instance, the unique optimal solution is given by the
clusters $\mathcal{C} = \curlyof{C_1, \ldots, C_k}$, and outliers
$Z$. Let $S' = \curlyof{c_1', \ldots, c_k'}$ be the optimal set of
centers. Let $\optradius{d'}$ denote the optimal radius of
$\inst'$. We will construct an alternate solution (clustering and
outliers) for $\inst'$ with cost at most $\optradius{d'}$. This
contradicts the uniqueness of the optimal solution, and thus fails to
satisfy the definition of perturbation resilience.

The following claim is also easy to establish (refer \reflemma{opt-unchanged-undirected}).
\begin{claim}
$\optradius{d'} = \optradius{d}$
\end{claim}

Next, we show the existence of an alternate solution of cost at most
$\optradius{d'}$. Consider the set of outliers $Z' = Z \setminus
\setof{q}$. Let $\mathcal{C'}$ be the Voronoi partition of $V
\setminus Z'$ induced by $S'$. Clearly the clustering $\mathcal{C'}$
is different from $\mathcal{C}$. Further, since $d'(c_i, q) \leq
\optradius{d} = \optradius{d'}$, we have, $\cost{d'}{\mathcal{C'}, S';
  Z'} = \max_{u \in V \setminus Z'} d'(S', u) \leq
\optradius{d'}$. This contradicts the uniqueness of the optimal
clustering and outliers of $\inst'$.
\end{proof}

\begin{lemma}\labellemma{kco-q-nonoutlier}
Consider any point $p \in V \setminus Z$, let $p \in C_i$. For all $q \in C_j $,
$d(p, q) > \optradius{d}$.
\end{lemma}
\begin{proof}
Follows from the fact, that instance $(V \setminus Z, d, k)$ is a $2$-perturbation
resilient instance for \kcenter and \reflemma{kc-sep}.
\end{proof}

\reflemma{kco-sep} follows immediately from \reflemma{kco-q-outlier} and \reflemma{kco-q-nonoutlier}.
\subsubsection{Proof of \reflemma{kco-out-ball}}
Let $C_i$ be the smallest cardinality cluster. Assume for the sake of
contradiction that the claim is false, i.e., there exists $p \in Z$,
such that $\sizeof{\ball{d}{p, 2 \cdot \optradius{d}} \Inter Z} \geq n_i$.

We construct a distance function $d'$ which is a metric $2$-perturbation of $d$. Consider the complete graph $G$ with edge lengths
$\ell$. Let $E' = \curlyof{(p, v): v \in \ball{d}{p, 2 \cdot \optradius{d}} \Inter Z}$. The edge lengths are defined as follows:
\[
\ell(u, v) = 
  \begin{cases} 
    \min \curlyof{d(u, v), \optradius{d}}  \quad& (u,v) \in E' \\
    d(u, v) \quad& \text{otherwise}
  \end{cases}
\]
For any pair of points $u, v$, the distance $d'(u, v)$ is the shortest path distance between
$u$ and $v$ in graph $G$, using $\ell$. 
Note that, for all $u, v \in V$, $\ell(u, v) \geq \frac{d(u, v)}{2}$.
We can
immediately make the following observation about $d'$.
\begin{obs}\labelobs{kcenter-obs-3}
$d'$ is a metric $2$-perturbation of $d$.
\end{obs}

Consider the instance $\inst' = (V, d', k, z)$. Since, $\inst'$ is a $2$-perturbed instance, the optimal
clustering and outliers are $\mathcal{C} = \curlyof{C_1, \ldots, C_k}$ and $Z$ respectively. Let $S' = \curlyof{c_1', \ldots, c_k'}$ be 
the optimal set of centers inducing the $\mathcal{C}$. Further, $\optradius{d'}$ denotes the cost of optimal solution. Again note that $\optradius{d'} = \optradius{d}$.

Now, consider the set of outliers $Z'' = \parof{Z \setminus \ball{d}{p, 2 \cdot \optradius{d}}} \bigcup C_i$. Let $S'' = S' \setminus \setof{c_i} \Union \setof{p}$ be a set of $k$ centers, and $\mathcal{C''}$ is a Voronoi partition of $V \setminus Z''$ induced by $S''$. 
For any point $u \in C_{\ell}$, where $\ell \neq i$, 
$d'(S'', u) \leq d'(c_\ell', u) \leq \optradius{d'}$. For any point $u \in \ball{d}{p, 2 \cdot \optradius{d}} \Inter Z$, we have,
$d'(S'', u) \leq d'(p, u) \leq \optradius{d} = \optradius{d'}$. Therefore, for any point $u \in V \setminus Z''$,
$d'(S'', u) \leq \optradius{d'}$. This implies $\cost{d'}{\mathcal{C''}, S''; Z''} \leq \optradius{d'}$. Clearly the clustering $\mathcal{C''}$ is different
from $\mathcal{C}$. Since $Z \cap C_i = \emptyset$, therefore $\sizeof{Z''} = \sizeof{Z} - \sizeof{\ball{d}{p, 2 \cdot \optradius{d}} \Inter Z} + \sizeof{n_i} \leq z$. Thus, $\mathcal{C''}, Z''$ is another solution for instance $\inst'$ having cost at most the optimal. In other words, the optimal solution of $\inst'$ is not unique, and this leads 
to contradiction.

\begin{figure}[!tbp]
  \centering
  \includegraphics[width=0.5\textwidth]{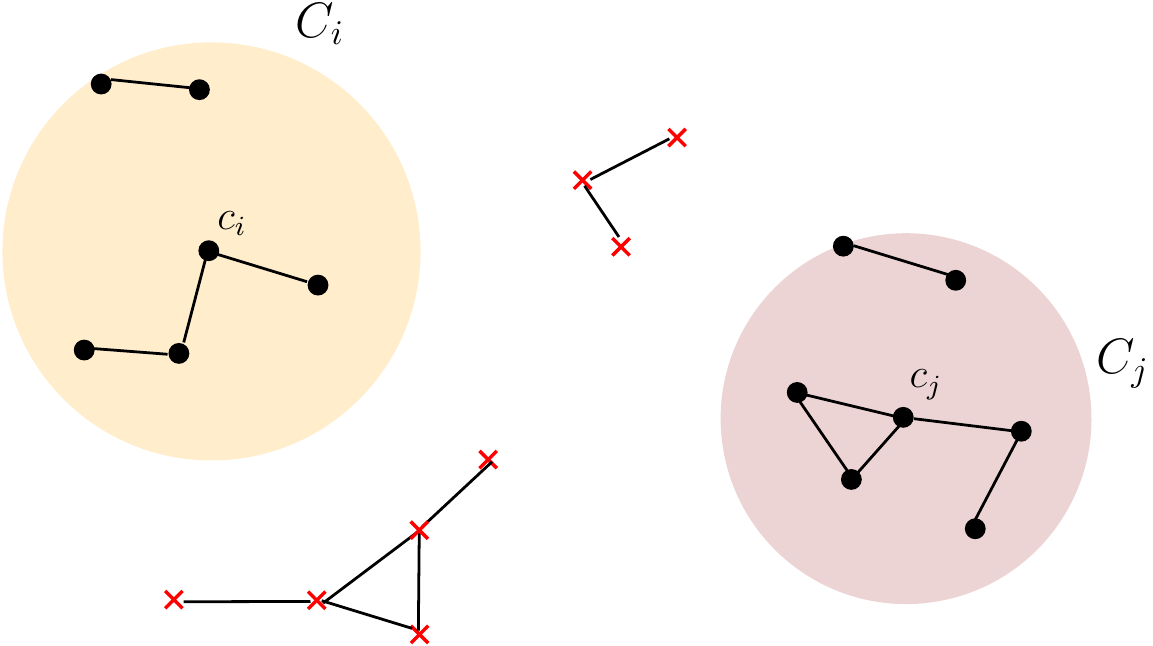}
  \caption{Graph $G_R$ corresponding to $R < \optradius{d}$ in $2$-perturbation resilient \kcenterout instance}
  \labelfigure{kco_3}
\end{figure}

\subsection{Integrality Gap and Proof of \reftheorem{kco-lp-integrality}}
In this section, we show that \refequation{kco-LP} is infeasible for
$R < \optradius{d}$. Recall in \reflemma{kco-sep}, we showed that the
optimal clusters are well-separated from each other and also from the
outliers. Therefore, in graph $G_R$, the connected components are
either subsets of optimal clusters or outliers (See
\reffigure{kco_3}). As a consequence, in a fractional solution,
non-outlier points can only be covered by points inside the cluster,
and similarly outliers can be covered by outliers only. However,
unlike \kcenter, here the tricky part is, the fractionally open
outliers can potentially cover a lot of points. We 
show that this in fact is not possible because of the sparsity of
an outlier's neighborhood. 

Suppose the claim is not true, that is for some $R < \optradius{d}$,
\refequation{kco-LP} has a feasible solution $(x^*, y^*)$. Let
$\mathcal{C} = \curlyof{C_1, \ldots, C_k}$ be the set of clusters and
$Z$ be the outliers in the unique optimal solution of $\inst$.

First, let us consider the simpler case when $y^*(Z) = 0$. Recall
\reflemma{kco-sep}, for every $p \in C_i$, ($i \in [k]$), the distance
to any $q \notin C_i$ is more than $\optradius{d}$. In other words,
for any $p \in C_i$, $\nbr{p} \subseteq C_i$, and for any $w \in Z$,
$\nbr{w} \bigcap V \setminus Z = \emptyset$. Therefore, $y^*(Z) = 0$
and the LP constraint $x_{uv} = 0, \forall v \in V, u \notin \nbr{v} $
implies (1) for any $w \in Z$, $x^*_{uw} = 0$ for all $u \in V$; (2)
for any $v \in V \setminus Z$, and $w \in Z$, $x^*_{wv} =
0$. Therefore, $(x^*, y^*)$ [restricted to $V \setminus Z$] is a
feasible fractional solution for \refequation{kc-LP} defined for the
\kcenter instance $\inst' = (V \setminus Z, d, k)$, and parameter
$R$. The optimal radius of $\inst'$ is also $\optradius{d}$. Therefore
by \reftheorem{kc-lp}, we cannot have a feasible fractional solution
for $R < \optradius{d}$, leading to a contradiction.

We focus on the case $y^*(Z) > 0$. Without loss of generality
assume that the optimum clusters are numbered such that 
$n_1 \le n_2 \le \ldots \le n_k$.
For $i \in [k]$ let $a_i = y(C_i)$
and let $b = y^*(Z)$. For a point $p$ let $\cov_p = \sum_{u} x^*_{up}$ be
the amount to which $p$ is covered. For a set of points $S$
we let $\cov(S)$ denote $\sum_{p \in S} \cov_p$.

\begin{claim}
  Total coverage of outlier points, that is, $\cov(Z) = \sum_{p \in Z} \rho_p < b n_1$.
\end{claim}
\begin{proof}
  Recall that an outlier point can only be covered by an outlier point.
  Further a point $q$ can cover point $p$ only if $p$ is in the ball of
  radius $R$ around $q$. Thus we have
  $$\sum_{p \in Z} \cov_p \le \sum_{q \in Z} \sizeof{\ball{d}{q,R}} \cdot y_q < n_1 \sum_{q \in Z} y_q = bn_1$$
  where we used \reflemma{kco-out-ball} to strictly upper bound
  $\sizeof{\ball{d}{q,R}}$ by $n_1$.
\end{proof}

\begin{claim}
  Let $C_i$ be an optimum cluster such that $a_i < 1$. Then
  $\cov(C_i) \le n_i a_i$.
\end{claim}
\begin{proof}
  Only points in $C_i$ can cover any given point $p \in C_i$. Therefore
  $\cov_p \le a_i$ for each $p \in C_i$, and hence $\cov(C_i) \le n_ia_i$.
\end{proof}

Let $A = \{i \in [k] \mid a_i < 1\}$ be the indices of the clusters
whose total $y$ value is strictly less than $1$. Since $y(V) = k$,
we have $b \le \sum_{i \in A} (1-a_i)$. 
Using the preceding two claims we have the following:
\begin{eqnarray*}
  \cov(V)  & = & \cov(Z) + \sum_{i \in A} \cov(C_i)  + \sum_{j \not \in A} \cov(C_i) \\
   & \le & \cov(Z) + \sum_{i \in A} \cov(C_i) +  \sum_{j \not \in A} n_j \\
   & \le & \cov(Z) + \sum_{i \in A} n_i a_i +  \sum_{j \not \in A} n_j \\
   & \le & \cov(Z) - \sum_{i \in A} n_i (1-a_i) +  \sum_{j =1}^k n_j \\
   & \le & \cov(Z) - \sum_{i \in A} n_i (1-a_i) + (n-z) \\
   & < & b n_1 - n_1 \sum_{i \in A} (1-a_i) + (n-z) \\
   & < & n-z.
\end{eqnarray*}

This is a contradiction to the feasibility of the LP solution.

\section{Algorithm for {\kmedianout} under Perturbation Resilience}\labelsection{kmo}
In this section, we present a dynamic programming based algorithm for
{\kmedianout}, which gives an optimal solution when the instance is
$2$-perturbation resilient. First, we prove some structural properties
of a $2$-OPR {\kmedianout} instance. They serve as the key ingredient
in showing that our algorithm will return exact solution for $2$-OPR
instances.

This section is essentially a straight forward extension of the ideas
in \cite{AngelidakisMM17} once the model is set up. In a sense
the model justifies the natural extension of the algorithm 
from \cite{AngelidakisMM17} to the outlier setting. 

\subsection{Properties of $2$-perturbation resilient {\kmedianout} instance}
Angelidakis et al. \cite{AngelidakisMM17} proved that in the optimal
clustering of a $2$-perturbation resilient {\kmedian} instance, every
point is closer to its assigned center than to any point in a
different cluster. In the optimal solution of {\kmedianout}, points
are not only assigned to clusters, some points are identified as
outliers as well. Here, we extend the result of \cite{AngelidakisMM17}
to show that the optimal solution of a $2$-OPR {\kmedianout} instance
satisfies the property: any non-outlier point is closer to its
assigned center than to any point outside the cluster.

\begin{lemma}\labellemma{kmo-sep}
  Consider a $2$-perturbation resilient \kmedianout instance $\inst =
  (V, d, k, z)$.  Let $\mathcal{C} = \curlyof{ C_1, \ldots C_k }$, and
  $Z$ be the unique optimal clustering and outliers resp.  Consider
  any point $p \in V \setminus Z$, and let $p \in C_i$. For all $q
  \notin C_i$, we have $d(c_i, p) < d(p, q)$.
\end{lemma}

To prove \reflemma{kmo-sep}, we split it into two cases: we show that
it holds true for (1) all outlier points $q$; (2) all non-outlier
points $q$ belonging to a different optimal cluster.

\begin{lemma}\labellemma{kmo-p-outlier}
  Consider a $2$-perturbation resilient \kmedianout instance $\inst =
  (V, d, k, z)$.  Let $\mathcal{C} = \curlyof{ C_1, \ldots C_k }$, and
  $Z$ be the unique optimal clustering and outliers resp.  Consider
  any point $p \in V \setminus Z$, and let $p \in C_i$. Then, for any
  outlier $q \in Z$, we have $d(c_i, p) < d(p, q)$.
\end{lemma}
\begin{proof}
  Let $S = \curlyof{c_1, \ldots, c_k}$ be an optimal set of centers,
  inducing $\mathcal{C}$.  Without loss of generality we assume, $p
  \neq c_i$, since $c_i, q$ being distinct points $d(c_i, q) > 0 =
  d(c_i, c_i)$.

  Assume for the sake of contradiction, the claim is false, that is,
  for some $q \in Z$, $d(p, c_i) \geq d(p, q)$.  To prove the
  contradiction, we construct a distance function $d'$, which is a
  metric $2$-perturbation of $d$. And show that in the instance thus
  constructed, the optimal solution is not unique --- that is there
  exists an optimal clustering and outliers different from
  $\mathcal{C}; Z$. This contradicts the definition of perturbation
  resilience.

  We define $d'$ as follows: consider the complete graph $G$ on
  vertices $V$. The edge lengths in graph $G$ are given by the
  function $\ell$, where for any edge $(u, v)$,
\[
\ell(u, v) = 
  \begin{cases} 
    d(c_i, p)  \quad& (u,v) = (c_i, q) \\
    d(u, v) \quad& \text{otherwise}
  \end{cases}
\]
For any pair of points $u, v$, the distance $d'(u, v)$ is the shortest
path distance between $u$ and $v$ in graph $G$, using $\ell$. We can
make some simple observations about $d'$:
\begin{obs}\labelobs{kmed-obs-1}
$d'$ has the following properties:
\begin{enumerate}[label=\roman*)]
\item for any $u, v \in V$, 
\begin{align*}
d'(u, v)  
&=\min \curlyof{\ell(u, v), \ell(u, q) + \ell(q, c_i) + \ell(c_i, v), \ell(u, c_i) + \ell(c_i, q) + \ell(q, v)} \\
&= \min \curlyof{ d(u, v), d(u, q) + d(c_i, p) + d(c_i, v) , d(u, c_i) + d(c_i, p) + d(q, v)}
\end{align*}
\item $d'(c_i, p) = d'(c_i, q) = d(c_i, p)$. 
\end{enumerate}
\end{obs} 
\begin{obs}\labelobs{kmed-obs-2}
$d'$ is a metric $2$-perturbation of $d$. 
\end{obs}
\begin{proof}
  By triangle inequality, $d(c_i, q) \leq d(c_i, p) + d(p, q) \leq 2
  \cdot d(c_i, p)$ --- the last inequality follows from our
  assumption.  Therefore $\frac{d(c_i, q)}{2} \leq d(c_i, p) =
  \ell(c_i, q)$ Further, note that $d(c_i, p) < d(c_i, q)$. Indeed, as
  otherwise we can swap $p$ and $q$ in the optimal solution,
  i.e. identify $p$ as an outlier and assign $q$ to the nearest center
  in $S$. Therefore, $\frac{d(c_i, q)}{2} \leq \ell(c_i, q) < d(c_i,
  q)$. For any other edge $(u, v)$, $\ell(u, v) = d(u, v)$. As we
  claimed in \reflemma{d-perturb}, $d'$ defined as the shortest path
  metric on an undirected graph $G$ with edge lengths $\ell$
  satisfying the property $\frac{d(u, v)}{2} \leq \ell(u, v) \leq d(u,
  v)$, is a metric $2$-perturbation of $d$.
\end{proof}

Consider the instance $\inst' = (V, d', k, z)$. Since, $\inst'$ is a $2$-perturbation of $\inst$ instance, 
the unique optimal
solution is given by the clusters 
$\mathcal{C} = \curlyof{C_1, \ldots, C_k}$, and outliers $Z$. Let $S' = \curlyof{c_1', \ldots, c_k'}$ be 
an optimal set of centers. We show that we can construct an alternate
 solution of cost at most the optimal solution cost 
 by swapping $q$ with a non-outlier point. To this end we consider two case:
\begin{description}[style=unboxed,leftmargin=0cm]
\item [Case 1: $c_i' = c_i$.] Consider a solution for $\inst'$, with set of outliers
$Z' = Z \setminus \setof{q} \bigcup \setof{p} $, and centers $S'$. Let 
$\mathcal{C'}$ be a Voronoi partition of $V \setminus Z'$ induced by $S'$.
The cost of the clustering $C'$ is,
\begin{align*}
\cost{d'}{\mathcal{C'}, S'; Z'} 
&= \sum_{u \in V \setminus Z'} d'(S', u) 
= d'(S', q) + \sum_{u \in C_i \setminus \setof{p}} d'(c_i', u) + \sum_{\substack{s = 1 \\s \neq i}}^k \sum_{u \in C_s} d'(c_s', u) \\
&\leq d'(c_i', q) - d'(c_i', p) + \sum_{s = 1}^k \sum_{u \in C_s} d'(c_s', u)  \\
&= d'(c_i, q) - d'(c_i, p) + \cost{d'}{\mathcal{C'}, S'; Z}  \quad (\because c_i' = c_i) \\
&= \cost{d'}{\mathcal{C}, S'; Z} \quad (\because d'(c_i, p) = d'(c_i, q) \text{ by \refobs{kmed-obs-1}} )
\end{align*}
\item [Case 2: $c_i' \neq c_i$.] We can assume without loss of
  generality, $S' \setminus \setof{c_i'} \bigcup \setof{c_i}$ is not
  an optimal set of centers for $\inst'$, otherwise the the argument
  is same as Case 1.  In particular, this implies, $\sum_{u \in C_i}
  d'(c_i, u) > \sum_{u \in C_i} d'(c_i', u)$.  Recall, for any two
  points $u, v \in V$, $d'(u, v) \leq d(u, v)$. We claim there must be
  a point $r \in C_i$, such that $d'(c_i', r) < d(c_i', r)$ (that is
  the distance between $c_i'$ and $r$ becomes strictly smaller after
  perturbation).  Indeed this is true, as otherwise,
\begin{align*}
\sum_{u \in C_i} d'(c_i', u) = \sum_{u \in C_i} d(c_i', u) \geq \sum_{u \in C_i} d(c_i, u) \geq \sum_{u \in C_i} d'(c_i, u)
\end{align*}
 where the first inequality uses the fact that $c_i$ is a center in the optimal solution of $\inst$.
 Now, $d'(c_i', r) < d(c_i', r)$ implies couple of things: (1) $c_i' \neq r$, 
 as in that case $d'(c_i', r) = d(c_i', r) = 0$; 
 (2) $d'(c_i', r) = \min \{ \ell(r, q) + \ell(q, c_i) + \ell(c_i, c_i'), \ell(r, c_i) + \ell(c_i, q) + \ell(q, c_i')\}$. 
Also, $d'(q, c_i') = \min \{ \ell(q, c_i'),  \ell(q, c_i) + \ell(c_i, c_i')\}$. Putting it together, 
we get $d'(c_i', r) \geq d'(c_i', q)$.

Consider a solution for $\inst'$, with set of outliers
$Z' = Z \setminus \setof{q} \bigcup \setof{r} $, and centers $S'$. Let 
$\mathcal{C'}$ be a Voronoi partition of $V \setminus Z'$ induced by $S'$.
The cost of the solution is,
\begin{align*}
\cost{d'}{\mathcal{C'}, S'; Z'} 
&= \sum_{u \in V \setminus Z'} d'(S', u) 
= d'(S', q) + \sum_{u \in C_i \setminus \setof{p}} d'(c_i', u) + \sum_{\substack{s = 1 \\s \neq i}}^k \sum_{u \in C_s} d'(c_s', u) \\
&\leq d'(c_i', q) - d'(c_i', r) + \sum_{s = 1}^k \sum_{u \in C_s} d'(c_s', u) \quad (\because c_i' = c_i) \\
&\leq \cost{d'}{\mathcal{C}, S'; Z} \quad (\because d'(c_i', r) \geq d'(c_i', q)) 
\end{align*}
\end{description}
In both cases, we constructed a solution for $\inst'$ which is
different from the optimal solution $\mathcal{C}; Z$, and has cost
less than or equal to the optimal cost. This contradicts the
uniqueness of the optimal solution.
\end{proof}

Next we show that \reflemma{kmo-sep} holds true for all non-outliers
points $q$ belonging to an optimal cluster different from $C_i$. The
proof is same as the one given in \cite{AngelidakisMM17}, we briefly
sketch it here for completeness.

\begin{lemma}\labellemma{kmo-p-center}
  Consider a $2$-perturbation resilient \kmedianout instance $\inst =
  (V, d, k, z)$.  Let $\mathcal{C} = \curlyof{ C_1, \ldots C_k }$, and
  $Z$ be the unique optimal clustering and outliers resp. Let $S =
  \setof{c_1, \ldots, c_k}$ be optimal centers inducing $\mathcal{C}$.
  Let $p \in V \setminus Z$ be an arbitrary point, and $c_i$ be the
  center it is assigned to. For any other center $c_j$ ($c_j \neq
  c_i)$, it follows $2 \cdot d(p, c_i) < d(p, c_j)$.
\end{lemma}
\begin{sketch}
Suppose the claim is not true, that is, for some $c_j \neq c_i$, $2 \cdot d(p, c_i) \geq d(p, c_j)$.
Similar to \reflemma{kmo-p-outlier}, we construct a distance 
function $d'$ which is a metric $2$-perturbation of $d$.
To this end, consider the complete graph $G$ defined on the vertex set $V$, 
with edge lengths $\ell$, where (1) $\ell(c_j, p) = d(c_i, p)$; (2) for every other edge 
$(u, v)$, $\ell(u, v) = d(u, v)$. We define $d'$, as
the shortest path distance (using $\ell$) between vertices in graph $G$. 
\begin{obs} \labelobs{kmed-obs-3}
$d'$ has the following properties:
\begin{enumerate}[label=\roman*)]
\item for any $u, v \in V$, such that $(u, v) \neq (c_j, p), (c_i, p)$, 
\begin{align*}
d'(u, v)  
&=\min \{ \ell(u, v), \ell(u, c_j) + \ell(c_j, p) + \ell(p, v), \ell(u, p) + \ell(p, c_j) + \ell(c_j, v)\} \nonumber \\
&= \min \{ d(u, v), d(u, c_j) + d(c_i, p) + d(p, v) , d(u, p) + d(c_i, p) + d(c_j, v)\}
\end{align*}
\item $d'(c_j, p) = d'(c_i, p) = d(c_i, p)$
\item $d'$ is a metric $2$-perturbation of $d$
\end{enumerate}
\end{obs}
Since instance $\inst$ is $2$-OPR for {\kmedianout}, even for the
perturbed instance $\inst' = (V, d', k, z)$ the unique optimal
clustering is $\mathcal{C}$ and outliers is $Z$. We can further show
that for any two points $u, v \in C_i$ (and $C_j$), $d'(u, v) = d(u,
v)$.  Thus $c_i$, $c_j$ are cluster centers in the optimal solution of
$\inst'$. Now, consider a solution for $\inst'$ with clustering
$\mathcal{C'} = \mathcal{C} \setminus \setof{C_i, C_j} \bigcup
\setof{C_i \setminus \setof{p}, C_j \bigcup \setof{p}}$.  and outliers
$Z$.  We can show that the cost of this solution $C'; Z$ is at most
the cost of the optimal solution $\mathcal{C}; Z$. Thus contradicting
the fact that the optimal solution is unique.
\end{sketch}

\begin{corollary}\labelcorollary{kmo-p-nonout}
Consider any point $p \in V \setminus Z$, and let $C_i$ be the optimal cluster $p$ is assigned to. Then, for
any other point $q$ from a different cluster $C_j$ ($i \neq j$), $d(p, c_i) < d(p, q)$.
\end{corollary}

\reflemma{kmo-sep} follows immediately from \refcorollary{kmo-p-nonout} and \reflemma{kmo-p-outlier}.

\subsection{Algorithm}
In the previous section, we showed that in the optimal solution of a
$2$-perturbation resilient {\kmedianout} instance, any non-outlier
point is closer to its assigned center than to any point outside the
cluster. This gives a nice structure to the optimal solution. In
particular, the optimal clusters form subtrees in the minimum spanning
tree over input point set.  We leverage this property to design a
dynamic programming based algorithm to identify the optimal clusters
and outliers.  In what follows, we interchangebly use the terms point
and vertex.

\begin{lemma}\labellemma{kmo-tree}
  Let $\inst = (V, d, k, z)$ be a $2$-perturbation resilient instance
  of the {\kmedianout} problem. Let $T$ be a minimum spanning tree on
  $V$. The optimal clusters of $\inst$, $C_1, \ldots C_k$ are subtrees
  in $T$ i.e. for any two points $p, q \in C_i$, all the points along
  the unique tree path between $p$, and $q$ belongs to cluster $C_i$.
\end{lemma}
\begin{proof}
  Let $c_i$ denote the center of cluster $C_i$. To prove the lemma, it
  is sufficient to show that every point on the unique tree path
  between $p$ and $c_i$ belongs to cluster $C_i$.  We prove this via
  induction on the length of path between $p$ and $c_i$. Let $u$ be
  the vertex after $p$ along this path. Since $(p, u)$ is an MST edge,
  we have $d(p, u) \leq d(p, c_i)$. By \reflemma{kmo-sep}, $u$ must
  belong to $C_i$. The proof then follows by applying induction on $u$
  to $c_i$ path.
\end{proof}

\reflemma{kmo-tree} implies that we can find the optimal solution of
$\inst$ by solving the following optimization problem, which we call
\tp: Partition the MST $T$ into $k$ subtrees $P_1, \ldots, P_k$, with
centers $c_1, \ldots c_k$ (each $c_i \in P_i$) and identify remaining
$Z$ vertices of the tree as outliers, where $\sizeof{Z} \leq z$. The
goal is to minimize the following objective function,
\begin{align*}
\sum_{i = 1}^k \sum_{u \in P_i} d(c_i, u)
\end{align*}


Solving \tp on a general tree is complicated. We simplify it by
transforming $T$ into a binary tree $T'$ with dummy vertices. The
procedure is as follows: while there is a vertex $v$ with more than
two children, pick any two children of $v$ --- $v_1$, and $v_2$;
create a new child (\emph{dummy vertex}) $u$ of $v$; reattach subtrees
rooted at $v_1$, and $v_2$ as children of $u$. At the end of this
process, let $U$ be the set of dummy vertices added.  For each dummy
vertex $u \in U$, set $d(u, v) = 0$, for every $v \in U \bigcup V$.

Now consider the following optimization problem (\btp): Partition binary tree 
$T'$ into $k$ subtrees $P_1', \ldots, P_k'$,
with centers $c_1', \ldots c_k'$ (each $c_i' \in P_i' \bigcap V$) and 
identify remaining $Z'$ vertices of the tree as outliers, where $\sizeof{Z' \bigcap V} \leq z$.
The cost function we want to minimize is,
\begin{align*}
\sum_{i = 1}^k \sum_{u \in P_i'} d(c_i, u)
\end{align*}

It is not hard to show, that given a solution to \tp, we can construct
a solution for \btp of equal cost, and vice-versa. Thus, to solve
{\kmedianout} it is sufficient to solve \btp on the binary tree $T'$
with dummy nodes. Given an optimal solution $P_1', \ldots, P_k'; Z'$
for \btp, the optimal clusters of the corresponding \kmedianout
instance is $P_1' \bigcap V, \ldots, P_k' \bigcap V$ and outliers is
$Z' \bigcap Z$.

To optimally solve \btp we use dynamic programming. For the rest of
the section, we consider $T$ to be the input binary tree, with $V$
being the vertices corresponding to points, and $U$ denotes the dummy
vertices.  Let $T_u$ denote the subtree rooted at $u$. Further let
$\ell_u, r_u$ respectively denote the left child, right child of $u$.

Let $\dpcost{u, j, t, c}$ be the minimum cost of partitioning the
points in subtree $T_u$ into $j$ clusters after discarding $t$ points
as outliers. Here $c$ can be any vertex in $V$ or it can be the null
(denoted using $\emptyset$).  The clustering satisifies the following
constraints:
\begin{itemize}
\item if $c = \emptyset$, then $u$ is marked as an outlier. 
\item if $c \neq \emptyset$, then the cluster in which $u$ belongs has center $c$.
\item Each cluster forms a subtree in $T_u$.
\end{itemize}
We can define $\dpcost{u, j, t, c}$ using the following recursive formula. 
\begin{description}
\item[$c = \emptyset, u \in V$.] Here $u$ is an outlier. Hence, $\ell_u$ and $r_u$ are assigned to
centers $c' \in T_{\ell_u}$ and $c'' \in T_{r_u}$ respectively. Further, since $u$ is already being
marked as an outlier, there can be $t-1$ outliers between $T_{\ell_u}$ and $T_{r_u}$.
\begin{align*}
&\dpcost{u, j, t, c} 
= \min \left\{ \dpcost{\ell_u, j', t', c'} + \dpcost{r_u, j'', t'', c''} : \right.\nonumber\\
&\qquad\qquad \left. {} j' + j'' = j, t' + t'' = t - 1, c' \in T_{\ell_u} \bigcup \emptyset, c'' \in T_{r_u} \bigcup \emptyset \right\}
\end{align*}
\item[$c = \emptyset, u \notin V$.] Here $u$ is an outlier. However, since it is a dummy vertex 
we do not count it as one of $t$ outliers in $T_u$. 
\begin{align*}
&\dpcost{u, j, t, c} 
= \min \left\{ \dpcost{\ell_u, j', t', c'} + \dpcost{r_u, j'', t'', c''} : \right.\nonumber\\
&\qquad\qquad \left. {} j' + j'' = j, t' + t'' = t, c' \in T_{\ell_u} \bigcup \emptyset, c'' \in T_{r_u} \bigcup \emptyset \right\}
\end{align*}
\item[$c \notin T_{\ell_u} \bigcup T_{r_u}$.] The recursive formula is defined by $4$ cases (lines 1-4 in the formula). 
The explanation for each case is as follows: 
(1) Neither $l_u$ nor $r_u$ is assigned to the same cluster as $u$. They are either outliers,
or they are assigned to centers $ c', c''$ in subtree $T_{\ell_u}, T_{r_u}$ resp. 
(2) $r_u$ is assigned to the same cluster as $u$
but not $\ell_u$. It is either an outlier or assigned to a center $ c' \in T_{\ell_u}$ 
(3) $\ell_u$ is assigned to the same cluster as $u$
but not $r_u$. It is either an outlier or assigned to a center $ c'' \in T_{r_u}$
(4) Both $\ell_u$ and $r_u$ are assigned to the same cluster as $u$.
\begin{align}
&\dpcost{u, j, t, c} 
= d(u, c) + \min \Big (\nonumber\\
&\qquad\min \left\{ \dpcost{\ell_u, j', t', c'} + \dpcost{r_u, j'', t'', c''} : \right.\nonumber\\
	&\qquad\qquad \left. {} j' + j'' = j - 1, t' + t'' = t, c' \in T_{\ell_u} \bigcup \emptyset, c'' \in T_{r_u}  \bigcup \emptyset \right\}, \\
&\qquad\min \left\{\dpcost{\ell_u, j', t', c'} + \dpcost{r_u, j'', t'', c} : \right.\nonumber\\
	&\qquad\qquad \left. {} j' + j'' = j, t' + t'' = t, c' \in T_{\ell_u} \bigcup \emptyset \right\}, \\
&\qquad\min \left\{ \dpcost{\ell_u, j', t', c} + \dpcost{r_u, j'', t'', c''} : \right.\nonumber\\
	&\qquad\qquad \left. {} j' + j'' = j, t' + t'' = t, c'' \in T_{r_u} \bigcup \emptyset \right\}, \\
&\qquad\min \left\{ \dpcost{\ell_u, j', t', c} + \dpcost{r_u, j'', t'', c} : \right.\nonumber\\
	&\qquad\qquad \left. {} j' + j'' = j - 1, t' + t'' = t \right\} \Big )
\end{align}
\item[$c \in T_{\ell_u}$.] The recursive formula in this case is obtained by removing lines (1), (2) from
the above formula.
\item[$c \in T_{\ell_u}$.] The recursive formula in this case is obtained by removing lines (1), (3) from
the above formula. 
\end{description}

\begin{remark}
The algorithm we presented easily generalizes to give exact solution for
$2$-perturbation resilient instances of other clustering with outliers problems 
like \kcenterout, \kmeansout, and more general $\ell_p$
objectives.
\end{remark}

\noindent{\bf Acknowledgements:} CC thanks Mohit Singh for initial discussions 
on the integrality of the LP relaxation for $2$-perturbation-resilient
instances of \kmedian. We thank Yury Makarychev for comments on
Voronoi clustering for \kcenter.

\bibliographystyle{plain}%
\bibliography{perturbation_resilient}

\appendix
\section{Omitted proofs from \refsection{prelim}}\labelappendix{gen}

\subsection{Proof of \reflemma{voronoi-kc}}
Since, $\sizeof{\mathcal{C'}} = k-1$,
there must be a cluster $C_t' \in \mathcal{C'} $, such that $C_t' \bigcap S \geq 2$.
Let $c_i, c_j \in S$ be the cluster centers which belong to $C_t'$. Wlog,
assume $c_t'$, the center of cluster $C_t'$ does not belong to cluster $C_j$.

\begin{description}[style=unboxed,leftmargin=0cm]
\item [Case 1: $V \setminus S' \nsubseteq C_j$]. Consider any point 
$q \in V \setminus \parof{S' \bigcup C_j}$. Consider the set of $k$ centers
$S'' = S' \Union \setof{q}$. Let $\mathcal{C''}$ be a corresponding Voronoi partition. Clearly 
$\cost{d}{\mathcal{C''}, S''} \leq \cost{d}{\mathcal{C'}, S'} \leq \optradius{d}$, as
adding a new center can not increase the clustering cost. Now, for any 
$c \in S'' \setminus \curlyof{c_t', q}$, we have $d(c, c_j) \geq d(c_t', c_j)$. Therefore,
in the Voronoi partition $\mathcal{C''}$, we can assume wlog, either $c_t'$ and $c_j$
are in the same cluster, or $q$ and $c_j$ are in same cluster. However both $c_t', q \notin C_j$.
Thus, $\mathcal{C''}$ is a different $k$ clustering of $V$ of cost at most  $\optradius{d}$.

\item [Case 2: $V \setminus S' \subseteq C_j$]. In this case, we have $S' = S \setminus \setof{c_j}$. Further,
for any $\ell \neq j$, $\sizeof{C_\ell} = 1$. Therefore, there exists a point $q \in C_j$,
such that $d(c_j, q) = \optradius{d}$. Since, $\cost{d}{\mathcal{C'}, S'} \leq \optradius{d}$, we have,
$d(S', q) \leq \optradius{d}$. Therefore, there exists a different Voronoi partition $\mathcal{C''}$ 
induced by $S = S' \bigcup \setof{c_j}$, where the points $q$ and $c_j$ do not belong in the same cluster. 
\end{description}

\subsection{Proof of \reflemma{d-perturb}}
Since $d'$ is defined as the shortest path distance (over non-negative edge lengths) 
in graph $G$, it satisfies triangle inequality.
Also, as mentioned in the lemma statement $\ell(u, v) \leq d(u, v)$. Therefore $d'(u, v) \leq \ell(u, v) \leq d(u, v)$.
Consider any two points $u, v \in V$. Let $P = (u, u_1, \ldots, v)$ be an arbitrary directed 
$u \leadsto v$ path in graph $G$. The length of path $P$ is given by $\ell(P) = \ell(u, u_1) + \ldots + 
\ell(u_t, v) \geq 1/2 \cdot \parof{d(u, u_1) + \ldots + d(u_t, v)} \geq \frac{d(u, v)}{2}$. Here the
last inequality uses the fact that $d$ satisfies triangle inequality. 
Therefore, $d'(u, v) = \Min_{P} \ell(P) \geq \frac{d(u, v)}{2}$. 

\subsection{Proof of \reflemma{opt-unchanged}}
Since for any pair of points $u, v \in V$, we have $d'(u, v) \leq d(u, v)$, clearly 
$\optradius{d'} \leq \optradius{d}$. Let $C_t$ be the largest radius optimal cluster in $\inst$,
i.e., $\max_{u \in C_t} d(c_t, u) = \optradius{d}$. Therefore, for every $c \in C_t$, there exists a point
$r(c) \in C_t$, such that $d(c, r(c)) \geq \optradius{d}$.
Now, for any path $P$ between $c$ and $r(c)$ in graph $G$, which does not include an edge from $E'$,
$\ell(P) = \sum_{e \in P} \ell(e) = \sum_{e \in P} d(e) \geq d(c, r(c)) \geq \optradius{d}$, by triangle inequality. Further,
for any path $P'$ between $c$ and $r(c)$ in graph $G$, which includes atleast one edge from $E'$,
$\ell(P') = \sum_{e \in P' \bigcap E'} \ell(e) + \sum_{e \in P' \setminus E'} \ell(e) \geq 
\sum_{e \in P' \bigcap E'} \min \curlyof{\ell(e), \optradius{d}} + \sum_{e \in P' \setminus E'} \ell(e) \geq \min \curlyof{\ell(P'), \optradius{d}} \geq \optradius{d}$. Therefore, $d'(c, r(c)) \geq \optradius{d}$. Now recall, we assumed $\mathcal{C}$ is an optimal clustering in $\inst'$, then $C_t$ is an optimal cluster in $\inst'$
Therefore,
\begin{align*}
\optradius{d'} = \cost{d'}{\mathcal{C}, S'} 
\geq \min_{c \in C_t} \max_{u \in C_t} d'(c, u)
\geq \min_{c \in C_t} d'(c, r(c))
\geq \optradius{d}
\end{align*}
Therefore, $\optradius{d'} = \optradius{d}$.

\bigskip
\noindent The proofs of \reflemma{d-perturb-undirected}, and \reflemma{opt-unchanged-undirected} are similar to the above.

\end{document}